\documentclass[acmsmall]{ec19acm}

\usepackage{booktabs} 

\usepackage[ruled]{algorithm2e} 

\SetAlFnt{\small}
\SetAlCapFnt{\small}
\SetAlCapNameFnt{\small}
\SetAlCapHSkip{0pt}
\IncMargin{-\parindent}

\usepackage{bbm}

\newtheorem{thm}{Theorem}[section]
\newtheorem{thm*}{Theorem}[section]
\newtheorem{cor}{Corollary}[section]
\newtheorem{cor*}{Corollary}[section]
\newtheorem{definition}{Definition}[section]
\newtheorem{example}{Example}[section]
\newtheorem{remark}{Remark}[section]
\newtheorem{prop}{Proposition}[section]
\newtheorem{prop*}{Proposition}[section]
\newtheorem{property}{Property}[section]
\newtheorem{assumption}{Assumption}

\newcommand{\bX}{\boldsymbol{X}}
\DeclareMathOperator*{\argmax}{arg\,max}

\newcommand{\Rb}{\mathbb{R}}
\newcommand{\Ex}{\mathbb{E}}

\newcommand{\wbX}{\widetilde{\bX}}

\AtBeginDocument{%
  \providecommand\BibTeX{{%
    \normalfont B\kern-0.5em{\scshape i\kern-0.25em b}\kern-0.8em\TeX}}}

\copyrightyear{2019} 
\acmYear{2019} 
\setcopyright{acmcopyright}
\acmConference[EC '19]{ACM EC '19: ACM Conference on Economics and Computation}{June 24--28, 2019}{Phoenix, AZ, USA}
\acmBooktitle{ACM EC '19: ACM Conference on Economics and Computation (EC '19), June 24--28, 2019, Phoenix, AZ, USA}
\acmPrice{15.00}
\acmDOI{10.1145/3328526.3329589}
\acmISBN{978-1-4503-6792-9/19/06}

\begin{document}

\title{A Marketplace for Data: An Algorithmic Solution}

\author{Anish Agarwal}
\email{anish90@mit.edu}
\affiliation{%
  \institution{Massachusetts Institute of Technology}
}

\author{Munther Dahleh}
\email{dahleh@mit.edu}
\affiliation{%
  \institution{Massachusetts Institute of Technology}
}

\author{Tuhin Sarkar}
\email{tsarkar@mit.edu}
\affiliation{%
  \institution{Massachusetts Institute of Technology}
}

\renewcommand{\shortauthors}{Agarwal, et al.}

\begin{abstract}
In this work, we aim to design a data marketplace; a robust real-time matching mechanism to efficiently buy and sell training data for Machine Learning tasks. While the monetization of data and pre-trained models is an essential focus of industry today, there does not exist a market mechanism to price training data and match buyers to sellers while still addressing the associated (computational and other) complexity. The challenge in creating such a market stems from the very nature of data as an asset: (i) it is freely replicable; (ii) its value is inherently combinatorial due to correlation with signal in other data; (iii) prediction tasks and the value of accuracy vary widely; (iv) usefulness of training data is difficult to verify a priori without first applying it to a prediction task. As our main contributions we: (i) propose a mathematical model for a two-sided data market and formally define the key associated challenges; (ii) construct algorithms for such a market to function and analyze how they meet the challenges defined. We highlight two technical contributions: (i) a new notion of ``fairness" required for cooperative games with freely replicable goods; (ii) a truthful, zero regret mechanism to auction a class of combinatorial goods based on utilizing Myerson's payment function and the Multiplicative Weights algorithm. These might be of independent interest.


\end{abstract}

\ccsdesc[500]{Theory of computation~Algorithmic game theory and mechanism design}

\keywords{Data Marketplaces, Value of Data, Shapley Value, Online Combinatorial Auctions}


\maketitle

\newpage
	\section{Introduction}
	\label{sec: Introduction}
	\noindent \textbf{A Data Marketplace - Why Now?} 
Machine Learning (ML) is starting to take the place in industry that "Information Technology" had in the late 1990s: businesses of all sizes and in all sectors, are recognizing the necessity to develop predictive capabilities for continued profitability. To be effective, ML algorithms rely on high-quality training data -- however, obtaining relevant training data can be very difficult for firms to do themselves, especially those early in their path towards incorporating ML into their operations. This problem is only further exacerbated, as businesses increasingly need to solve these prediction problems in real-time (e.g. a ride-share company setting prices, retailers/restaurants sending targeted coupons to clear inventory), which means data gets ``stale" quickly. Therefore, we aim to design a data marketplace -- a real-time market structure for the buying and selling of training data for ML.

\medskip 
 \noindent \textbf{What makes Data a Unique Asset?} (i) Data can be replicated at zero marginal cost -- in general, modeling digital goods (i.e., freely replicated goods) as assets is a relatively new problem (cf. \cite{aiello2001priced}). (ii) Its value to a firm is inherently combinatorial i.e., the value of a particular dataset to a firm depends on what other (potentially correlated) datasets are available - hence, it is not obvious how to set prices for a collection of datasets with correlated signals. (iii) Prediction tasks and the value of an increase in prediction accuracy vary widely between different firms - for example, a 10\% increase in prediction accuracy has very different value for a hedge fund maximizing profit compared to a logistics company trying to decrease inventory costs. (iv) The authenticity and usefulness of data is difficult to verify a priori without first applying it to a prediction task - continuing the example from above, a particular dataset of say satellite images may be very predictive for a specific financial instrument but may have little use in forecasting demand for a logistics company (and this is infeasible to check beforehand).

\medskip 
\noindent \textbf{Why Current Online Markets Do Not Suffice?} Arguably, the most relevant real-time markets to compare against are: (i) online ad auctions (cf. \cite{varian2009online}); (ii) prediction markets (cf. \cite{wolfers2004prediction}). Traditionally, in these markets (e.g. online ad auctions) the commodity (e.g. ad-space) is not a replicable good and buyers have a strong prior (e.g. historical click-through-rate) on the value of the good sold (cf. \cite{liu2006designing, Zhang:2014:ORB:2623330.2623633}). In contrast for a data market, \textit{it is infeasible for a firm to make bids on specific datasets as it is unlikely they have a prior on its usefulness}. Secondly, it is infeasible to run something akin to a second price auction (and variants thereof) since data is freely replicable (unless a seller artificially restricts the number of replications, which may be suboptimal for maximizing revenue). This problem only gets significantly more complicated due to the combinatorial nature of data.  \textit{Thus any market which matches prediction tasks and training features on sale, needs to do so based on which datasets collectively are, empirically the most predictive and ``cheap" enough for a buyer}. This is a capability online ad markets and prediction markets do not currently have. See Section \ref{sec:appendix_literature_review} for a more thorough comparison with online ad and prediction markets. 

\subsection{Overview of Contributions} 
\noindent \textbf{Mathematical Model of Two-Sided Data Market. Formal Definition of Key Challenges.} As the main contribution of this paper, we propose a mathematical model of a system design for a data marketplace; we rigorously parametrize the participants of our proposed market - the buyers, the sellers and the marketplace itself (Sections \ref{sec:seller_definition}, \ref{sec:buyer_definition}, \ref{sec:marketplace_definition}) - and the mechanism by which they interact (Section \ref{sec:market_dynamics}). 
{\em This is a new formulation, which lays out a possible architecture for a data marketplace, and takes into account some of the key properties that make data unique}; 
it is freely replicable, it is combinatorial (i.e., features have overlapping information), buyers having no prior on usefulness of individual datasets on sale and the prediction tasks of buyers vary widely. In Section \ref{sec:marketplace_properties}, we study the key challenges for such a marketplace to robustly function in real-time, which include: (i) incentivizing buyers to report their internal valuations truthfully; (ii) updating the price for a collection of correlated datasets such that revenue is maximized over time; (iii) dividing the generated revenue ``fairly" among the training features so sellers get paid for their marginal contribution; (iv) constructing algorithms that achieve all of the above and are efficiently computable (e.g. run in polynomial time in the parameters of the marketplace)?

\medskip 
\noindent \textbf{Algorithmic Solution. Theoretical Guarantees.} In Section \ref{sec:market_construction}, we construct algorithms for the various functions the marketplace must carry out: (i) allocate training features to and collect revenue from buyers; (ii) update the price at which the features are sold; (iii) distribute revenue amongst the data sellers.  In Section \ref{sec:main_results}, we prove these particular constructions do indeed satisfy the desirable marketplace properties laid out in Section \ref{sec:marketplace_properties}. We highlight two technical contributions: (i) Property \ref{prop:replication_robustness}, a novel notion of ``fairness" required for cooperative games with freely replicable goods, which generalizes the standard notion of Shapley fairness; (ii) a truthful, zero regret mechanism for auctioning a particular class of combinatorial goods based on utilizing Myerson's payment function (cf. \cite{myerson1981optimal}) and the Multiplicative Weights algorithm (cf. \cite{mw}). These might be of independent interest.

\subsection{Motivating Example from Inventory Optimization} \label{sec:motivating_retail_example}
We begin with an example from inventory optimization to help build intuition for our proposed architecture for a data marketplace (see Section \ref{sec:model} for a mathematical formalization of these dynamics). We refer back to this example throughout the paper as we introduce various notations and algorithmic constructions.
\medskip

\noindent \textbf{Inventory Optimization Example}: Imagine data sellers are retail stores selling anonymized minute-by-minute foot-traffic data streams into a marketplace and data buyers are logistics companies who want features (i.e. various time series) that best forecast future inventory demand. In such a setting, even though a logistics company clearly knows there is predictive value in these data streams on sale, it is reasonable to assume that the company does not have a good prior on what collection of foot-traffic data streams are most predictive for demand forecasting, and within their budget. Thus, practically speaking, such a logistics company cannot make accuracy, individual bids for each data stream (this is even without accounting for the significant additional complication arising from the overlap in signal i.e., the correlation that invariably will exist between the foot-traffic data streams of the various retail stores). 

Instead what a logistics company does realistically have access to is a well-defined cost model for not predicting demand well (cf. \cite{heyman2004stochastic, bullwhipdemand}) - e.g., ``10\% over/under-capacity costs \$10,000 per week''. Hence it can make a bid into a data market of what a marginal increase in forecasting accuracy of inventory demand is worth to it - e.g. ``willing to pay \$1000 for a percentage increase in demand forecasting accuracy from the previous week". 

In such a setting, the marketplace we design performs the following steps: 
\begin{enumerate}
\item The logistics company supplies a prediction task (i.e., a time series of historical inventory demand) and a bid signifying what a marginal increase in accuracy is worth to it
\item The market supplies the logistics company with foot-traffic data streams that are ``cheap" enough as a function of the bid made and the current price of the data streams
\item A ML model is fit using the foot-traffic data streams sold and the historical inventory demand
\item Revenue is collected based {\em only on the increased accuracy in forecasting inventory demand} \footnote{Model evaluation could potentially be done on an out-of-sample test set or based on actual prediction performance on future unseen demand}
\item Revenue is divided amongst all the retail stores who provided foot-traffic data
\item The price associated with the foot-traffic data streams is then updated
\end{enumerate}
What we find especially exciting about this example is that it can easily be adapted to a variety of commercial settings. Examples include: (i) hedge funds sourcing alternative data to predict certain financial instruments; (ii) utility companies sourcing electric vehicle charging data to forecast electricity demand during peak hours; (iii) retailers sourcing online social media data to predict customer churn. 

Thus we believe the {\em dynamic described above can be a natural, scalable way for businesses to source data for ML tasks, without knowing a priori what combination of data sources will be useful}.

\subsection{Literature Review} \label{sec:appendix_literature_review}

\noindent \textbf{Auction design and Online Matching}. In this work, we are specifically concerned with online auction design in a two--sided market. There is a rich body of literature on optimal auction design theory initiated by~\cite{myerson1981optimal},~\cite{riley1981optimal}. We highlight some representative papers. In~\cite{rochet2003platform} and~\cite{caillaud2003chicken}, platform design and the function of general intermediary service providers for such markets is studied; in~\cite{gomes2014optimal}, advertising auctions are studied; in the context of ride--sharing such as those in Uber and Lyft, the efficiency of matching in~\cite{chen2016dynamic} and optimal pricing in~\cite{banerjee2015pricing} are studied. An extensive survey on online matching, in the context of ad allocation, can be found in~\cite{mehta2013online}. These paper generally focus on the tradeoff between inducing participation and extracting rent from both sides. Intrinsic to such models is the assumption that the value of the goods or service being sold is known partially or in expectation. This is the key issue in applying these platform designs for a data marketplace; as stated earlier, it is unrealistic for a buyer to know the value of the various data streams being sold a priori (recall the inventory example in Section \ref{sec:motivating_retail_example} in which a logistic company cannot realistically make accurate bids on separate data streams or bundles of data streams). Secondly these prior works do no take into account the freely replicable, combinatorial nature of a good such as data.

\medskip 
\noindent \textbf{Online Ad Auctions}. See \cite{varian2009online} for a detailed overview. There are two key issues with online ad markets that make it infeasible for data. Firstly, ad-space is not a replicable good i.e., for any particular user on an online platform, at any instant in time, only a single ad can be shown in an ad-space. Thus an \textit{online ad market does not need to do any ``price discovery"} - it simply allocates the ad-space to the highest bidder; to ensure truthfulness, the highest bidder pays the second highest bid i.e., the celebrated second price auction (and variants thereof). In contrast, for a freely replicable good such as data, a second price auction does not suffice (unless a seller artificially restricts a dataset to be replicated a fixed number of times, which may be suboptimal for maximizing revenue). Secondly, buyers of online ad-space have a strong prior on the value of a particular ad-space - for example, a pharmaceutical company has access to historical click-through rates (CTR) for when a user searches for the word ``cancer ". So it is possible for firms to make separate bids for different ad-spaces based on access to past performance information such as CTR (cf. \cite{liu2006designing, Zhang:2014:ORB:2623330.2623633}). In contrast, since prediction tasks vary so greatly, past success of a specific training feature on sale has little meaning for a firm trying to source training data for its highly specific ML task; again, making it is infeasible for a firm to make bids on specific datasets as they have no prior on its usefulness. 

\medskip
\noindent \textbf{Prediction Markets}. Such markets are a recent phenomenon and have generated a lot of interest, rightly so. See \cite{wolfers2004prediction} for a detailed overview. Typically in such markets, there is a discrete random variable, $W$, that a firm wants to accurately predict. The market executes as follows: (i) ``experts" sell probability distributions $\Delta_W$ i.e., predictions on the chance of each outcome; (ii) the true outcome, $w$, is observed; (iii) the market pays the ``experts" based on $\Delta_W$ and $w$. In such literature, payment functions based on the Kullback--Leibler divergence are commonly utilized, as they incentivize ``experts" to be truthful (cf. \cite{hanson2012logarithmic}). Despite similarities, prediction markets remain infeasible for data as ``experts" have to explicitly choose which tasks to make specific predictions for. In contrast, it is not known a priori whether a particular dataset has any importance for a prediction task; in the inventory optimization example in Section \ref{sec:motivating_retail_example}, retail stores selling foot-traffic data cannot realistically know which logistics company's demand forecast their data stream will be predictive for (again, this is only exacerbated when taking into account the overlap of information between features). A data market must instead provide a real-time mechanism to match training features to prediction tasks based on the increase in predictive value from the allocated features.

\medskip
\noindent \textbf{Information Economics}. There has been an exciting recent line of work that directly tackles data as an economic good which we believe to be complimentary to our work. We divide them into three major buckets and highlight some representative papers: (i) data sellers have detailed knowledge of the specific prediction task and incentives to exert effort to collect high-quality data (e.g. reduce variance) are modeled~\cite{optmechsellinfo, accuracyforsale}; (ii) data sellers have different valuations for privacy and mechanisms that tradeoff privacy loss vs. revenue gain are modeled~\cite{sellingprivacyauction, surveyprivacycost}; (iii) studying the profitability of data intermediaries who supply consumer data to firms that want to sell the very same customers more targeted goods ~\cite{designpriceinfo}. These are all extremely important lines of work to pursue, but they focus on different (but complementary) objectives. 

Referring back to the inventory optimization example in Section \ref{sec:motivating_retail_example}), we model the sellers (retail stores) as simply trying to maximize revenue by selling foot-traffic data they already collect. Hence we assume they have {\em (i) no ability to fundamentally increase the quality of their data stream; (ii) no knowledge of the prediction task; (iii) no concerns for privacy}. In many practical commercial settings, these assumptions do suffice as the data is sufficiently anonymized, and these sellers are trying to monetize data they are already implicitly collecting through their operations. We focus our work on such a setting, where firms are trying to buy features to feed their ML models, and believe our formulation to be the most relevant for it. It would be interesting future work to find ways of incorporating privacy, feedback and the cost of data acquisition into our model.

	\section{The Model - Participants and Dynamics} 
	\label{sec:model}

\subsection{Sellers}\label{sec:seller_definition} Let there be $M$ sellers, each supplying data streams in this marketplace. We formally parameterize a seller through the following single quantity: 

\medskip

\noindent \textbf{Feature.} $X_j \in \mathbb{R}^T, \  j \in [M]$ is a vector of length $T$. 

For simplicity, we associate with each seller a single feature and thus restrict $X_j$ to be in $\mathbb{R}^T$. Our model is naturally extended to the case where sellers are selling multiple streams of data by considering each stream as another ``seller" in the marketplace. We refer to the matrix denoting any subset of features as $\bX_S$, $S \subset [M]$. Recall from the motivations we provide in Sections \ref{sec:motivating_retail_example} and \ref{sec:appendix_literature_review} for our model, we assume data sellers do not have the ability to change the quality of the data stream (e.g. reducing variance) they supply into the market nor any concerns for privacy (we assume data is sufficiently anonymized as is common in many commercial settings). Additionally sellers have no knowledge of the prediction tasks their data will be used for and simply aim to maximize revenue from the datasets that they already have on hand.

\subsection{Buyers}\label{sec:buyer_definition}
Let there be $N$ buyers in the market, each trying to purchase the best collection of datasets they can afford in this marketplace for a particular prediction task. We formally parameterize a buyer through the following set of quantities, for $n \in [N]$: 

\medskip
%

\noindent \noindent \textbf{Prediction Task.} $Y_n  \in \mathbb{R}^T$ is a vector of $T$ labels that Buyer $n$ wants to predict well
\footnote{To reduce notational overload, we abstract away the partition of $Y_n$ into training and test data.}.

We provide a clarifying example of $Y_n$ and $X_j$, using the inventory optimization example in Section \ref{sec:motivating_retail_example}. There, the historical inventory demand for the logistics company is $Y_n$ and each historical foot-traffic data stream sold by retailers is $X_j$. The ``prediction task" is then to forecast inventory demand, $Y_n$, from time-lagged foot traffic data, $X_j$ for $j \in [M]$.

\medskip
\noindent \textbf{Prediction Gain Function.} $\mathcal{G}_n : \mathbb{R}^{2T} \to [0, 1]$, the prediction gain function, takes as inputs the prediction task $Y_n$ and an estimate $\hat{Y}_n$, and outputs the quality of the prediction. 

For regression, an example of $\mathcal{G}_n$ is $1 - \text{RMSE}$ \footnote{$\text{RMSE} = \frac{1}{Y_{\max} - Y_{\min}}\sqrt{\sum_{i=1}^T (\hat{Y_i} - Y_i)^2 / T
}$, where: (i) $\hat{Y_i}$ is the predicted value for $i \in [T]$ produced by the machine learning algorithm, $\mathcal{M}$, (ii) $Y_{\max},  Y_{\min}$ are the max and min of $Y_n$ respectively.} (root-mean-squared-error). For classification, an example of $\mathcal{G}_n$ is Accuracy \footnote{$\text{Accuracy} = \frac{1}{T}\sum_{i=1}^T \mathbbm{1}({\hat{Y_i} = Y_i})$, with $\hat{Y_i}$ defined similarly to that above}. In short, a larger value for $\mathcal{G}_n$ implies better prediction accuracy. To simplify the exposition (and without any loss of generality of the model), we assume that all buyers use the same gain function i.e., $\mathcal{G} = \mathcal{G}_n$ for all $n$.

\medskip
\noindent \textbf{Value of Accuracy.} $\mu_n  \in \mathbb{R}_{+}$ is how much Buyer $n$ values a \textit{marginal} increase in accuracy. 

As an illustration, recall the inventory optimization example in Section \ref{sec:motivating_retail_example} where a logistics company makes a bid of the form, ``willing to pay \$1000 for a percentage increase in demand forecasting accuracy from the previous week". We then have the following definition for how a buyer values an increase in accuracy, 

\begin{definition}\label{def:value_from_prediction}
Let $\mathcal{G}$ be the prediction gain function. We define the value Buyer $n$ gets from estimate $\hat{Y}_n$ as: 
\[
\mu_n \cdot \mathcal{G}(Y_n, \hat{Y}_n)
\]
i.e., $\mu_n$ is what a buyer is willing to pay for a unit increase in $\mathcal{G}$.
\end{definition}

\begin{remark}
Though a seemingly natural definition, we view it as one of the key modeling decisions we make in our design of a data marketplace. In particular, a buyer's valuation for data does not come from specific datasets, but rather from an increase in prediction accuracy of a quantity of interest. 
\end{remark}

\begin{remark}
A potential source of confusion is we require $\mu_n$ to be linear while many ML error metrics are non-linear. For example, in balanced, binary classification problems, randomly guessing labels has expected accuracy of $50\%$, which has zero value (and not $\mu_n / 2$). However, such non-linearities can easily be captured in the gain function, $\mathcal{G}$. For the balanced, binary classification problem, $\mathcal{G}$ can easily be normalized such that $50\%$-accuracy has value $0$ and $100\%$-accuracy has value 1 (specifically, $G=\frac{max(0, \widehat{\text{accuracy}})}{0.5}$). $\mu_n$ can thus be thought of as a buyer-specific scaling of how much they value an increase in accuracy. Indeed, linear utility models are standard in information economics (c.f. ~\cite{designpriceinfo}).
\end{remark}

\begin{remark}
Recall that to reduce notational overload, we let $Y_n$ refer to both test and train data. Specifically, $Y_n=(Y_n^{train},Y_n^{test})$. The ML-algorithm accesses $Y_n^{train}$ and the Gain function, $\mathcal{G}$ accesses $Y_n^{test}$ (i.e. $\hat{Y}_n$), as is standard in ML workflows.
\end{remark}

\smallskip
\noindent \textbf{Public Bid Supplied to Market.} $b_n \in \mathbb{R}_{+}$ is the public bid supplied to the marketplace. 

Note that $\mu_n$ is a private valuation. If Buyer $n$ is strategic, $\mu_n$ is not necessarily what is revealed to the marketplace. Thus we define $b_n$, which refers to the actual bid supplied to the marketplace (not necessarily equal to $\mu_n$).

\subsection{Marketplace}\label{sec:marketplace_definition} 

The function of the marketplace is to match buyers and sellers as defined above. As we make precise in Section \ref{sec:market_dynamics}, we model the $M$ sellers as fixed and the $N$ buyers as coming one at a time. We formally parameterize a marketplace through the following set of quantities, for $n \in [N]$: 

\medskip
\noindent \textbf{Price.} $p_n  \in \mathbb{R}_{+}$ is the price the marketplace sets for the features on sale when Buyer $n$ arrives.

As we make precise in Property \ref{prop:rev_maximizing}, we measure the quality of the prices $(p_1, \dots, p_N)$ set by the marketplace for each buyer by comparing against the optimal fixed price in hindsight (i.e., standard definition of regret). However, it is well-known that standard price update algorithms for combinatorial goods, which satisfy Property \ref{prop:rev_maximizing}, scale very poorly in $M$ (cf. ~\cite{regrethardenvyeasy}). Specifically, if we maintain separate prices for every data stream (i.e., if $p_n \in \mathbb{R}^M_{+}$)  it is easily seen that regret-minimizing algorithms such as Multiplicative Weights (cf. \cite{mw}) or Upper Confidence Bandits (cf. \cite{bandits}), will have exponential running time or exponentially loose guarantees (in $M$) respectively. In fact from~\cite{regrethardenvyeasy}, we know regret minimizing algorithms for even very simple non-additive buyer valuations are provably computationally intractable.

Thus to achieve a zero-regret price update algorithm, without making additional restrictive assumptions, we restrict $p_n$ to be a scalar rather than a $M$-dimensional vector. This is justified due to Definition \ref{def:value_from_prediction}, where we model a buyer's ``value for accuracy" (and the associated public bid) through the scalar, $\mu_n$ (and the scalar $b_n$ respectively). This allows the marketplace to control the quality of the predictions based on the difference between $p_n$ and $b_n$ (see Section \ref{sec:allocation_revenue_func} for details).

\medskip
\noindent \textbf{Machine Learning/Prediction Algorithm.} $\mathcal{M} : \mathbb{R}^{MT} \to \mathbb{R}^{T}$,  the learning algorithm utilized by the marketplace, takes as input the features on sale $\bX_M$, and produces an estimate $\hat{Y}_n$ of Buyer $n$'s prediction problem $Y_n$. 

$\mathcal{M}$ does not necessarily have to be supplied by the marketplace and is a simplifying assumption. Instead buyers could provide their own learning algorithm that they intend to use, or point towards one of the many excellent standard open-source libraries widely used such as SparkML, Tensorflow and Scikit-Learn (cf. \cite{MlLib2015, PedregosaSciKit2011, TensonFlow2016}) 
\footnote{Indeed a key trend in many business use cases, is that the ML algorithms used are simply lifted from standard open-source libraries. Thus the accuracy of predictions is primarily a function of the quality of the data fed to these ML algorithms.}. 

\medskip
\noindent \textbf{Allocation Function.} $\mathcal{AF} : (p_n, b_n; \bX_M) \to \wbX_M, \wbX_M \in \mathbb{R}^M$, takes as input the current price $p_n$ and the bid $b_n$ received, to decide the quality at which Buyer $n$ gets allocated the features on sale $\bX_M$ (e.g. by adding noise or subsampling the features). 

In Section \ref{sec:allocation_revenue_func}, we provide explicit instantiations of $\mathcal{AF}$ and detailed reasoning for why we choose this particular class of allocation functions.

\medskip
\noindent \textbf{Revenue Function.} $\mathcal{RF} : (p_n, b_n, Y_n; \mathcal{M}, \mathcal{G}, \bX_M) \to r_n, \ r_n \in \mathbb{R}_{+}$, the revenue function, takes as input the current price $p_n$, in addition to the bid and the prediction task provided by the buyer ($b_n$ and $Y_n$ respectively), to decide how much revenue $r_n$ to extract from the buyer.

\medskip
\noindent \textbf{Payment Division Function.} $\mathcal{PD} : (Y_n, \wbX_M; \mathcal{M}, \mathcal{G}) \to  \psi_n, \ \psi_n \in  [0, 1]^{M}$,  the payment-division function, takes as input the prediction task $Y_n$ along with the features that were allocated $\wbX_M$, to compute $\psi_n$, a vector denoting the marginal value of each allocated feature for the prediction task. 

\medskip
\noindent \textbf{Price Update Function.} $\mathcal{PF} : (p_n, b_n, Y_n; \mathcal{M}, \mathcal{G}, \bX_M) \to  p_{n+1}, \ p_{n+1} \in  \mathbb{R}_{+}$, the price-update function, takes as input the current price $p_n$, in addition to the bid and the prediction task provided by the buyer ($b_n$ and $Y_n$ respectively) to update the price for Buyer $n+1$. 

\subsubsection{Buyer Utility}\label{sec:buyer_utility}
We can now precisely define the utility function, $\mathcal{U}: \mathbb{R}_{+}  \times \mathbb{R}^{T}\to \mathbb{R}$, each buyer is trying to maximize,
\begin{definition}\label{def:buyer_utility}
The utility Buyer $n$ receives by bidding $b_n$ for prediction task $Y_n$ is given by 
\begin{align}\label{eq:buyer_utility}
\mathcal{U}(b_n, Y_n) \coloneqq \mu_n \cdot \mathcal{G}(Y_n, \hat{Y}_n) - \mathcal{RF}(p_n, b_n, Y_n)
\end{align}
where $\hat{Y}_n = \mathcal{M}(Y_n, \wbX_M)$ and $\wbX_M = \mathcal{AF}(p_n, b_n; \bX_M)$.
\end{definition}
In words, the first term on the right hand side (r.h.s) of \eqref{eq:buyer_utility} is the value derived from a gain in prediction accuracy (as in Definition \ref{def:value_from_prediction}). Note this is a function of the quality of the features that were allocated based on the bid $b_n$. The second term on the r.h.s of \eqref{eq:buyer_utility} is the amount the buyer pays, $r_n$. Buyer utility as in Definition \ref{def:buyer_utility} is simply the difference between these two terms.

\subsection{Marketplace Dynamics}\label{sec:market_dynamics}
We can now formally define the per-step dynamic within the marketplace (see Figure \ref{fig:market_architecture} for a graphical overview). Note this is a formalization of the steps laid out in the inventory optimization example in Section \ref{sec:motivating_retail_example}. When Buyer $n$ arrives, the following steps occur in sequence (we assume $p_0, b_0, Y_0$ are initialized randomly): 

\centerline{\rule{\linewidth}{.2pt}}
For $n \in [N]$:
\begin{enumerate}
	\item Market sets price $p_n$, where $p_n = \mathcal{PF}(p_{n-1}, b_{n-1}, Y_{n-1})$
	\item Buyer $n$ arrives with prediction task $Y_n$
	\item Buyer $n$ bids $b_n$ where $b_n = \argmax_{z \in \mathbb{R}_{+}} \mathcal{U}(z, Y_n)$
	\item Market allocates features $\wbX_M$ to Buyer $n$ , where $\wbX_M = \mathcal{AF}(p_n, b_n; \bX_M)$
	\item Buyer $n$ achieves $\mathcal{G}\Big(Y_n, \mathcal{M}(\wbX_M)\Big)$ gain in prediction accuracy
	\item Market extracts revenue, $r_n$, from Buyer $n$, where $r_n = \mathcal{RF}(p_n, b_n, Y_n; \mathcal{M}, \mathcal{G})$ 
	\item Market divides $r_n$ amongst allocated features using $\psi_n$, where $\psi_n = \mathcal{PD}(Y_n, \wbX_M; \mathcal{M}, \mathcal{G})$
\end{enumerate} 
\centerline{\rule{\linewidth}{.2pt}} 

\begin{figure}[]
	\centering
	\includegraphics[width=0.7\textwidth]{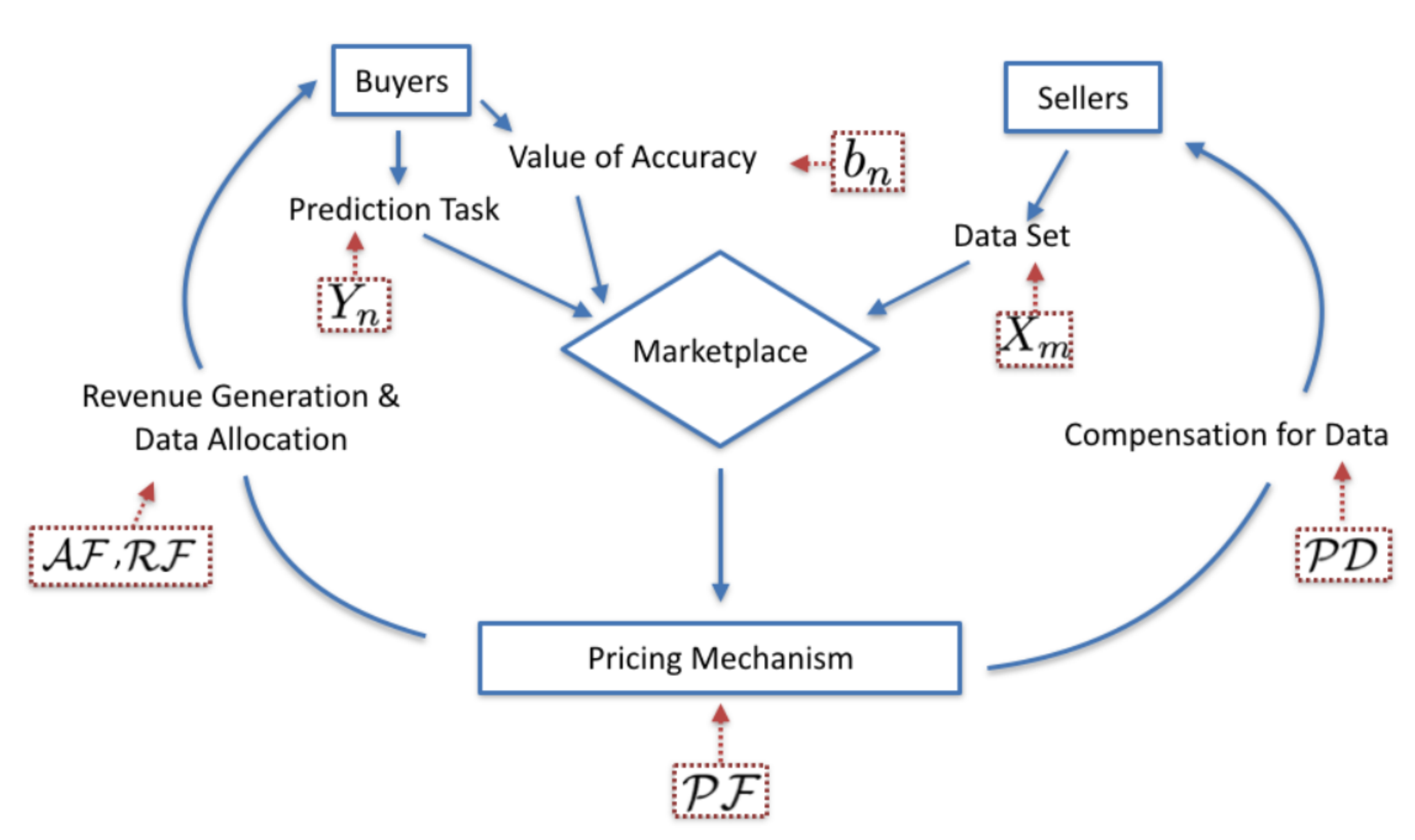}
	\caption{Overview of marketplace dynamics.}
	\label{fig:market_architecture}
\end{figure}

\begin{remark} \label{remark:data_is_reusable}
A particularly important (albeit implicit) benefit of the above proposed architecture is that the buyer's do not ever access the underlying features. Rather they only receive predictions through the ML model trained on the allocated features. This circumvents a known, difficult problem in designing data markets where sellers are reluctant to release potentially valuable data streams as they do not have control over who subsequently accesses it (since data streams are freely replicable).
\end{remark}

\begin{remark}\label{remark:central_price_setting}
In our proposed architecture, the price for each buyer is set centrally by the marketplace rather than by the sellers individually. A sellers simply supplies data streams to the marketplace and is assigned revenue based on the marginal contribution the data stream provides to the prediction task. Thus from the perspective of price setting, our model can equivalently be thought of as a single seller supplying multiple data streams to the market and adjusting $p_n$ to maximize overall revenue.
\end{remark}

\begin{remark} \label{remark:myopic_buyer}
We note from the dynamics laid out above (specifically Step 3), a buyer is ``myopic" over a single-stage i.e., Buyer $n$ comes into the market once and leaves after being provided the estimate $\hat{Y}_n$. Thus Buyer $n$ maximizing utility only over Step $n$. In particular, we do not study the additional complication if the buyer's utility is defined over multiple-stages.
\end{remark}

\begin{remark} \label{remark:externality}
Our proposed architecture does not take into account an important attribute of data; a firm's utility for a particular dataset may be heavily dependent on what other firms get access to it (e.g. a hedge fund might pay a premium to have a particularly predictive dataset only go to it). By modeling buyer's coming to the market one at a time, we do not study the externalities associated with a dataset being replicated multiple times.    
\end{remark}

	\section{Desirable Properties of Marketplace}
	\label{sec:marketplace_properties}
We define key properties for such a marketplace to robustly function in a large-scale, real-time setting, where buyers are arriving in quick succession and need to be matched with a large number of data sellers within minutes, if not quicker. Intuitively we require the following properties: (i) buyers are truthful in their bids; (ii) overall revenue is maximized; (iii) revenue is fairly divided amongst sellers; (iv) marketplace runs efficiently.  In Sections \ref{sec:property_truthful}-\ref{sec:comp_efficiency}, we formally define these properties. 

\subsection{Truthfulness} \label{sec:property_truthful}
\begin{property} [Truthful] \label{prop:truthful}
A marketplace is ``truthful" if  for all $Y_n$, 
\[
\mu_n = \argmax_{z \in \mathbb{R}_{+}} \mathcal{U}(z, Y_n)
\]
where $\mathcal{U}(z, Y_n)$ is defined as in Definition \ref{def:buyer_utility}.
\end{property}

Property \ref{prop:truthful} requires that the allocation function, $\mathcal{AF}$, and the revenue function, $\mathcal{RF}$, incentivize buyers to bid their true valuation for an increase in prediction accuracy. Note that we assume buyers do not alter their prediction task, $Y_n$.

\subsection{Revenue Maximization}\label{sec:rev_maximizaton_definition}
\begin{property} [Revenue Maximizing] \label{prop:rev_maximizing}
Let $\{(\mu_1, b_1, Y_1), (\mu_2, b_2, Y_2), \dots, (\mu_N, b_N, Y_N)\}$ be a sequence of buyers entering the market. A marketplace is ``revenue maximizing" if the price-update function, $\mathcal{PF}(\cdot)$, produces a sequence of prices, $\{p_1, p_2, \dots, p_n \}$, such that the ``worst-case"  average regret, relative to the optimal  price $p^*$ in hindsight, goes to $0$, i.e.,
\[
\lim_{N \to \infty} \dfrac{1}{N} \Big[ \sup_{\{(b_n, Y_n): n\in [N]\}} \Big( \sup_{p^* \in \mathbb{R}_{+}} \sum_{n=1}^{N} \mathcal{RF}(p^*, b_n, Y_n) - \sum_{n=1}^{N} \mathcal{RF}(p_n, b_n, Y_n) \Big) \Big].
\] 
\end{property}
As is convention, we term the expression with the square bracket as regret and denote it, $\mathcal{R}(N, M)$. Property \ref{prop:rev_maximizing} is the standard worst-case regret guarantee (cf. \cite{zinkevich03}). It necessitates the price-update function, $\mathcal{PF}$, produce a sequence of prices $p_n$ such that the average difference with the unknown optimal price in hindsight, $p^*$ goes to zero as $N$ increases. Note Property \ref{prop:rev_maximizing} must hold over the worst case sequence of buyers i.e, no distributional assumptions on $\mu_n, b_n, Y_n$ are made. 

\subsection{Revenue Division}
In the following section, we abuse notation and let $S \subset [M]$ refer to both the index of the training features on sale and to the actual features, $\bX_S$ themselves.
\subsubsection{Shapley Fairness}
\begin{property} [Shapley Fair] \label{prop:fair} 
A marketplace is ``Shapley-fair" if $\ \forall \ n \in [N], \forall \ Y_n$, the following holds on $\mathcal{PD}$ (and its output, $\psi_n$): 

\begin{enumerate}
\item \textbf{Balance}: $\sum_{m=1}^{M} \psi_n(m) = 1$
\item \textbf{Symmetry}: $\forall \ m, m' \in [M], \forall S \subset [M] \setminus \{m, m'\}$,  if $\mathcal{PD}(S \cup m, Y_n) = \mathcal{PD}(S \cup m',  Y_n)$, then $\psi_n(m) = \psi_n(m')$
\item \textbf{Zero Element}: $\forall \ m \in [M], \forall S \subset [M]$, if $\mathcal{PD}(S \cup m,  Y_n) = \mathcal{PD}(S,  Y_n)$, then $\psi_n(m) = 0$
\item \textbf{Additivity}:  Let the output of $\mathcal{PD}([M], Y^{(1)}_n), \mathcal{PD}([M], Y^{(2)}_n)$ be $\psi_n^{(1)}, \psi_n^{(2)}$ respectively. Let $\psi_n'$ be the output of $\mathcal{PD}([M], Y^{(1)}_n + Y^{(2)}_n)$. Then $\psi_n' = \psi_n^{(1)} + \psi_n^{(2)}$. 
\end{enumerate}
\end{property}

The conditions of Property \ref{prop:fair}, first laid out in \cite{shapley_allocation}, are considered the standard axioms of fairness. We choose them as they are the de facto method to assess the marginal value of goods (i.e., features in our setting) in a cooperative game (i.e., prediction task in our setting). 

\begin{remark}
A naive definition of the marginal value of feature $m$ would be a ``leave-one-out" policy, i.e., $\psi_n(m) = \mathcal{G}(Y_n, \mathcal{M}(\wbX_{[M]})) - \mathcal{G}(Y_n, \mathcal{M}(\wbX_{[M] \setminus {m}}))$. As the following toy example shows, the correlation between features would lead to the market ``undervaluing" each feature. Consider the simple case where there are two sellers each selling identical features. It is easy to see the ``leave-one-out" policy above would lead to zero value being allocated to each feature, even though they collectively might have great predictive value. This is clearly undesirable. That is why Property \ref{prop:fair} is a necessary notion of fairness as it takes into account the overlap of information that will invariably occur between the different features, $X_j$.
\end{remark}
We then have the following celebrated theorem from \cite{shapley_allocation},
\begin{thm}[Shapley Allocation]\label{thm:shapley_algorithm}
Let $\psi_{\text{shapley}} \in [0, 1]^{[M]}$ be the output of the following algorithm,
\begin{align}\label{eq:shapley_mechanism}
\psi_{\text{shapley}}(m) = \sum_{T \subset [M] \setminus \{m\}} \frac{|T|!(M - |T| - 1)!}{M!} \bigg( \mathcal{G}\Big(Y_n, \mathcal{M}(\wbX_{T \cup m})\Big) - \mathcal{G}\Big(Y_n, \mathcal{M}(\wbX_{T}) \Big) \bigg)
\end{align}
Then $\psi_{\text{shapley}}$ is the unique allocation that satisfies all conditions of Property \ref{prop:fair}
\end{thm}
Intuitively, this algorithm is computing the average marginal value of feature $m$ over all subsets $T \subset [M] \setminus \{m\}$. It is easily seen that the running time of this algorithm is $\Theta(2^M)$, which makes it infeasible at scale if implemented as is. But it still serves as a useful standard to compare against.

\subsubsection{Robustness to Replication}\label{sec:replication_robustness}
\begin{property} [Robustness to replication] \label{prop:replication_robustness} 
For all $m \in [M]$, let $m^+_i$ refer to the  $i^{th}$ replicated copy of $m$ i.e., $X^+_{m, i} = X_m$. Let $[M]^+ = \cup_m (m \cup_i m^+_i)$ refer to the set of original and replicated features.  Let $\psi^+_{n} = \mathcal{PD}([M]^+, Y_n)$.
Then a marketplace is $\epsilon$-``robust-to-replication" if $\ \forall \ n \in [N], \forall \ Y_n$, the following holds on $\mathcal{PD}$: 
\[
\psi^+_{n} (m) + \sum_{i}\psi^+_{n} (m^+_i) \le \psi_n(m) + \epsilon.
\]
\end{property}

Property \ref{prop:replication_robustness} is a novel notion of fairness, which can be considered a necessary additional requirement to the Shapley notions of fairness for freely replicable goods. We use Example~\ref{example:shapley_not_robust} below to elucidate how adverse replication of data can lead to grossly undesirable revenue divisions (see Figure \ref{fig:shapley_inadequate} for a graphical illustration). Note that implicit in the definition of Property \ref{prop:replication_robustness} is that the ``strategy-space" of the data sellers is the number of times they replicate their data.

\begin{example}\label{example:shapley_not_robust}
Consider a simple setting where the marketplace consists of only two sellers, A and B, each selling one feature which are both identical. By Property \ref{prop:fair}, the Shapley value of A and B are equal, i.e., $\psi(A) = \frac{1}{2},  \psi(B) = \frac{1}{2}$. However if seller A replicated her feature once and sold it again in the marketplace, it is easy to see that the new Shapley allocation will be $\psi(A) = \frac{2}{3},  \psi(B) = \frac{1}{3}$. Hence it is not robust to replication since the aggregate payment remains the same (no change in accuracy).
\end{example}

\begin{figure}[]
	\centering
	\includegraphics[width=0.45\textwidth]{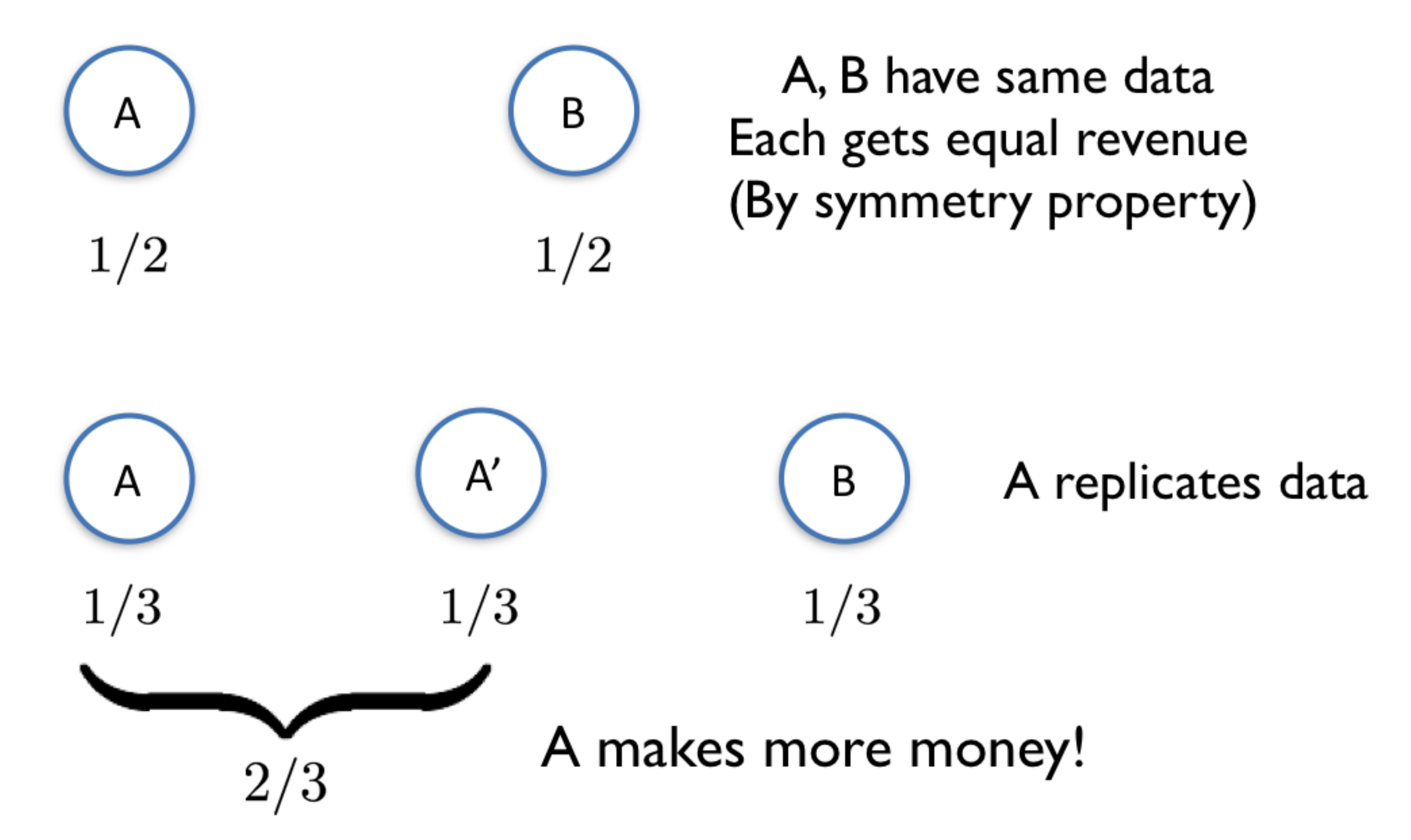}
	\caption{Shapley fairness is inadequate for freely replicable goods.}
	\label{fig:shapley_inadequate}
\end{figure}

Such a notion of fairness in cooperative games is especially important in modern day applications where: (i) digital goods are prevalent and can be produced at close to zero marginal cost; (ii) users get utility from bundles of digital goods with potentially complex combinatorial interactions between them. Two examples of such a setting are battery cost attribution among smartphone applications and reward allocation among ``experts" in a prediction market.

\subsection{Computational Efficiency}\label{sec:comp_efficiency}
We assume the Machine Learning algorithm, $\mathcal{M}$, and the Gain function, $\mathcal{G}$, each require computation running time of ${O}(M)$, i.e., computation complexity scales at most linearly with the number of features/sellers, $M$. We define the following computational efficiency requirement of the market,

\begin{property} [Efficient] \label{prop:efficient}
A marketplace is ``efficient" if for each Step $n$, the marketplace as laid out in Section \ref{sec:market_dynamics} runs in polynomial time in $M$, where $M$ is the number of sellers. In addition, the computational complexity for each step of the marketplace cannot grow with $N$.
\end{property}

Such a marketplace is feasible only if it functions in real-time. Thus, it is pertinent that the computational resources required for any Buyer $n$ to interface with the market are low i.e., ideally with run-time close to linear in $M$, the number of sellers, and not growing based on the number of buyers seen thus far. Due to the combinatorial nature of data, this is a non--trivial requirement as such combinatorial interactions normally lead to an exponential dependence in $M$; recall from earlier sections, the Shapley Algorithm in Theorem \ref{thm:shapley_algorithm} runs in $\Theta(2^M)$ and a naive implementation of Multiplicative Weights Algorithm for combinatorial goods runs in $\Theta(\exp(M))$.

	\section{Marketplace Construction} 
	\label{sec:market_construction}
We now explicitly construct instances of $\mathcal{AF}, \mathcal{RF}, \mathcal{PF}$ and $\mathcal{PD}$ and argue in Section \ref{sec:main_results} that the properties laid out in Section \ref{sec:marketplace_properties} hold for these particular constructions. 

\begin{remark}
In line with Remark \ref{remark:central_price_setting}, we can think of $\mathcal{AF}, \mathcal{RF}$ as instances of how to design a robust bidding, data allocation and revenue generation scheme from the buyer's perspective with the features sold held fixed (see Property \ref{prop:truthful}). Analogously, $\mathcal{PD}$ is a function for fair revenue division from the seller's perspective for a fixed amount of generated revenue (see Properties \ref{prop:fair} and \ref{prop:replication_robustness}). And $\mathcal{PF}$ is a function to centrally adjust the price of the features sold dynamically over time from the marketplace's perspective (see Property \ref{prop:rev_maximizing}).
\end{remark}

\subsection{Allocation and Revenue Functions (Buyer's Perspective)}\label{sec:allocation_revenue_func}
\textbf{Allocation Function.} 
Recall the allocation function, $\mathcal{AF}$, takes as input the current price $p_n$ and the bid $b_n$ received, to decide the quality of the features $\bX_M$, used for Buyer $n$'s prediction task.

Recall from Definition \ref{def:value_from_prediction} that a buyer's utility comes solely from the quality of the estimate $\hat{Y}_n$ received, rather than the particular datasets allocated. Thus the key structure we exploit in designing $\mathcal{AF}$ is that from the buyer's perspective, {\em instead of considering each feature $X_j$ as a separate good (which leads to computational intractability), it is greatly simplifying to think of $\bX_M$ as the total amount of ``information" on sale}. $\mathcal{AF}$ can thus be thought of as a function to {\em collectively} adjust the quality of all of $\bX_M$ based on the difference between $p_n$ and $b_n$ 

{\em Specifically, we choose $\mathcal{AF}$ to be a function that adds noise to/degrades $\bX_M$ proportional to the difference between $p_n$ and $b_n$}. This degradation can take many forms and  depends on the structure of $X_j$ itself. Below, we provide examples of commonly used allocation functions for some typical $X_j$ encountered in ML.

\begin{example}\label{example:alloc_function_real}
Consider $X_j \in \mathbb{R}^T$ i.e. sequence of real numbers. Then an allocation function (i.e. perturbation function), $\mathcal{AF}_1^*(p_n, b_n; X_j)$, commonly used (cf. \cite{accuracyforsale, sellingprivacyauction}) would be for $t$ in $[T]$,
\[
\tilde{X}_j(t) = X_j(t) + \max(0, p_n - b_n) \cdot \mathcal{N}(0, \sigma^2) 
 \]
where $\mathcal{N}(0, \sigma^2)$ is a univariate Gaussian.
\end{example}

\begin{example}\label{example:alloc_function_bits}
Consider $X_j \in \{0, 1\}^T$ i.e. sequence of bits. Then an allocation function (i.e. masking function), $\mathcal{AF}_2^*(p_n, b_n; X_j)$, commonly used (cf. \cite{adverseriallearningdata}) would be for $t$ in $[T]$,
\[
\tilde{X}_j(t) = B(t; \theta) \cdot X_j(t)   
 \]
 where $B(t; \theta)$ is an independent Bernoulli random variable with parameter $\theta = \min(\frac{b_n}{p_n}, 1)$.
\end{example}

In both examples if $b_n \geq p_n$, then the buyer is given $\bX_M$ as is without degradation. However if $b_n < p_n$, then $X_j$ is degraded in proportion to the difference between $b_n$ and $p_n$.

\begin{remark}\label{remark:future_AF_work}
Through Assumption \ref{assumption:robustness_to_noise} (see Section \ref{sec:key_assumptions}), we formalize a natural and necessary property required of any such allocation function so that Property \ref{prop:truthful} (truthfulness) holds. Specifically, for a fixed price $p_n$, increasing the bid $b_n$ cannot lead to a decrease in prediction quality. The space of possible allocations functions that meet this criteria is clearly quite large. We leave it as future work to study what is the optimal $\mathcal{AF}$ from the space of feasible allocation functions to maximize revenue.
\end{remark}

\begin{remark}
A celebrated result from \cite{myerson1981optimal} is that for single-parameter buyers, a single take-it-or-leave-it price for all data is optimal, i.e. if the bid is above the singe posted price then allocate all the data without any noise and if the bid is less than the price, allocate no data. However, maybe surprisingly, in our setting this result does not apply. This is due to an important subtlety in our formalism - while $\mu_n$ (how much a buyer values a marginal increase in accuracy), is a scalar, a buyer is also parametrized by $Y_n$, the prediction task. 

This leads to the following simple counter-example - imagine buyers are only of two types: (i) Type I with prediction task $Y_1$ and valuation $\mu_1$; (ii) Type II with prediction task $Y_2$ and valuation $\mu_2$. Further, let there be only two types of features on sale, $X_1$ and $X_2$. Assume $X_1$ is perfectly predictive of prediction task $Y_1$ and has zero predictive value for $Y_2$. Analogously, assume $X_2$ is perfectly predictive of $Y_2$ and has no predictive value for $Y_1$. Then it is easy to see that the optimal pricing mechanism is to set the price of $X_1$ to be $\mu_1$ and  $X_2$ to be $\mu_2$. Thus a single posted price is not optimal \footnote{When it comes to truthfulness (see Property \ref{prop:truthful}), the fact that we can parametrize a buyer’s valuation through a scalar, $\mu_n$, does indeed mean Myerson’s payment function (see \eqref{eq:myerson}) is truthful as long as the data allocation function is monotonic.}. 

More generally, if different datasets have varying amounts of predictive power for different buyer types, it is not even clear that a take-it-or-leave-it price per feature sold is optimal. 
\end{remark}

\medskip
\noindent \textbf{Revenue Function.} 
Recall from Definition \ref{def:value_from_prediction}, we parameterize buyer utility through the parameter $\mu_n$, i.e., how much a buyer values a marginal increase in prediction quality. This crucial modeling choice allows us to use Myerson's  payment function rule (cf. \cite{myerson1981optimal}) given below,
\begin{equation}\label{eq:myerson}
\mathcal{RF}^*(p_n, b_n, Y_n)  = b_n \cdot \mathcal{G}\bigg(Y_n, \mathcal{M} \Big(\mathcal{AF}^*(b_n, p_n) \Big) \bigg) - \int_{0}^{b_n} \mathcal{G}\bigg(Y_n, \mathcal{M} \Big(\mathcal{AF}^*(z, p_n) \Big) \bigg) dz.
\end{equation}
In Theorem \ref{thm:truthful}, we show that $\mathcal{RF}^*$ \textit{ensures Buyer $n$ is truthful} (as defined in Property \ref {prop:truthful}). Refer to Figure \ref{fig:allocation_and_myerson} for a graphical view of $\mathcal{AF}^*$ and $\mathcal{RF}^*$.

\begin{figure}[]
	\centering
	\includegraphics[width=0.8\textwidth]{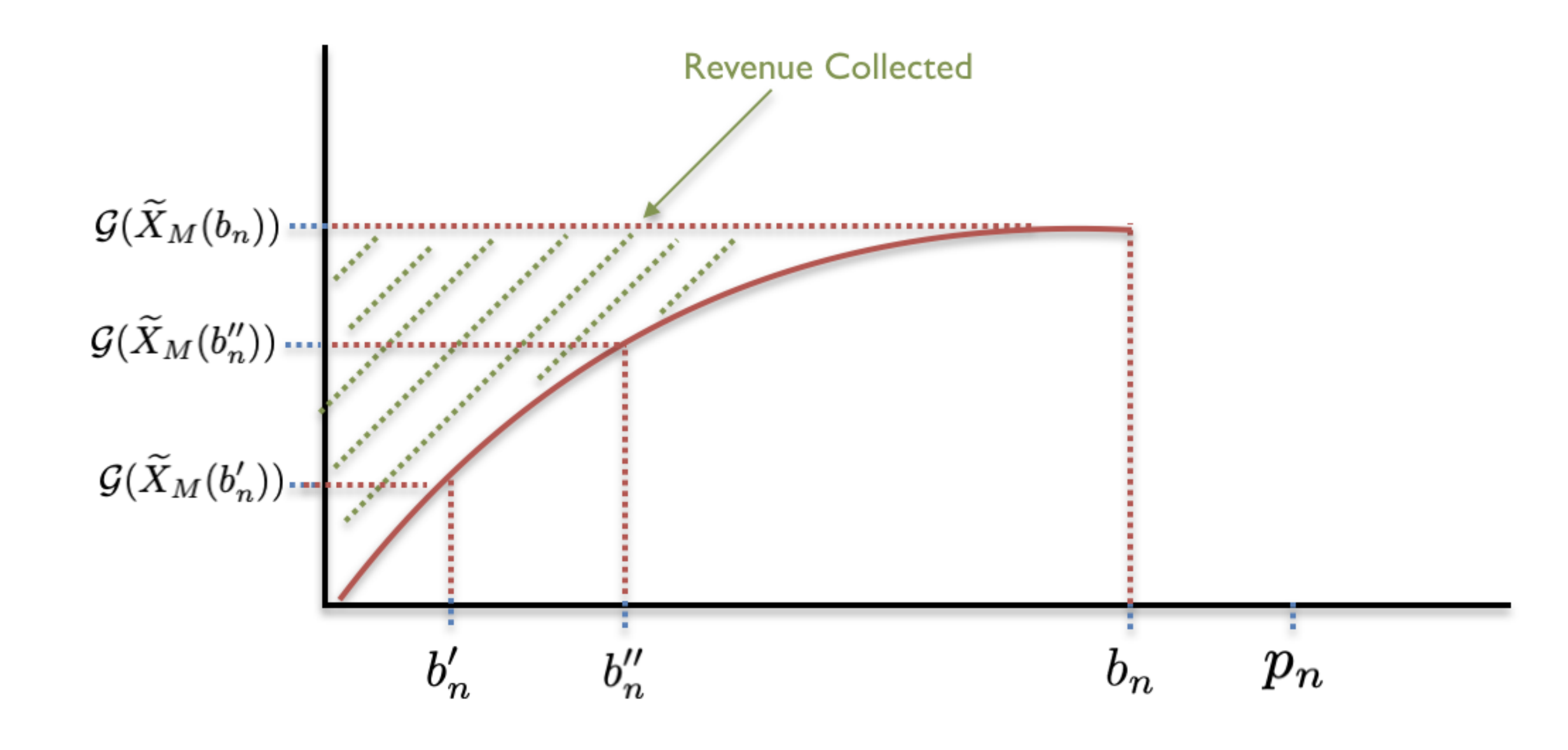}
	\caption{Features allocated ($\mathcal{AF}^*$) and revenue collected ($\mathcal{RF}^*$) for a particular price vector $p_n$ and bid $b_n$.}
	\label{fig:allocation_and_myerson}
\end{figure}

\subsection{Price Update Function (Marketplace Perspective)}\label{sec:price_update_func}
The market is tasked with how to pick $p_n$ for $n \in [N]$. Recall from Section \ref{sec:market_dynamics}, the market must decide $p_n$ before Buyer $n$ arrives (otherwise, it is easily seen that truthfulness cannot be gauranteed).

We now provide some intuition of how increasing/decreasing $p_n$ affects the amount of revenue collected, and the implicit tradeoff that lies therein. Observe from the construction of $\mathcal{RF}^*$ in \eqref{eq:myerson} that for a fixed bid, $b_n$, and prediction task $Y_n$, it is easily seen that if $p_n$ is picked to be too large, then the positive term in $\mathcal{RF}^*$ is small (as the degradation of the signal in $\bX_M$ is very high), leading to lower than optimal revenue generation. Similarly, if $p_n$ is picked to be too small, it is easily seen that the negative term in $\mathcal{RF}^*$ is large, which again leads to an undesired loss in revenue. 

In Algorithm \ref{alg-price-update-function}, we apply the Multiplicative Weights method to pick $p_n$ in an online fashion and to balance the tradeoff described above. To construct Algorithm \ref{alg-price-update-function} more precisely, we need to define some additional quantities. As we make precise in Assumption \ref{assumption:bounded_bids} in Section \ref{sec:main_results}, we assume bids come from some bounded set, $\mathcal{B} \subset \mathbb{R}_+$. We define $\mathcal{B}_{\max} \in \mathbb{R}$ to be the maximum element of $\mathcal{B}$ and $\mathcal{B}_{\text{net}}(\epsilon)$ to be a minimal $\epsilon$-net of $\mathcal{B}$
\footnote{We endow $\mathbb{R}$ with the standard Euclidean metric. An $\epsilon$-net of a set $\mathcal{B}$ is a set $K \subset \mathcal{B}$ such that for every point $x \in \mathcal{B}$, there is a point $x_0 \in K$ such that $|x - x_0| \leq \epsilon$.}. 
Intuitively, the elements of $\mathcal{B}_{\text{net}}(\epsilon)$ serve as our ``experts" (i.e. the different prices we experiment with) in the Multiplicative Weights algorithm. 

In Theorem \ref{thm:avg_regret}, we show that this algorithm does indeed achieve zero-regert with respect to the optimal $p^* \in \mathbb{R}_+$ in hindsight.
\begin{algorithm}[t]
	\SetAlgoNoLine
	\KwIn{$b_n, Y_n, \mathcal{B}, \epsilon, \delta$}
	\KwOut{$p_n$}
	Let $\mathcal{B}_{\text{net}}(\epsilon)$ be an $\epsilon$-net of $\mathcal{B}$; \\
	\For{$c^i \in \mathcal{B}_{\text{net}}(\epsilon)$ } 
	{
		Set $w^i_1 = 1$ \tcp*{initialize weights of all experts to 1}
	}
	\For{$n=1$ to $N$}
	{
		$W_n = \sum_{i=1}^{|\mathcal{B}_{\text{net}}(\epsilon)|} w^{i}_n$; \\
		Let $p_n = c^i$ with probability $w^i_n / W_n$ \tcp*{note $p_n$ is not a function of $b_n$}
		\For{$c^i \in \mathcal{B}_{\text{net}}(\epsilon)$}
		{
			Let $g^i_n = \mathcal{RF}^*(c^i, b_n, Y_n) / \mathcal{B}_{\max}$ \tcp*{revenue gain if price $c^i$ was used}
			Set $w^i_{n+1} = w^i_n \cdot (1 + \delta g^i_n)$ \tcp*{Multiplicative Weights update step}
		}
	}
	\caption{PRICE-UPDATE: $ \mathcal{PF}^*(b_n, Y_n, \mathcal{B}, \epsilon, \delta)$}
	\label{alg-price-update-function}
\end{algorithm}

\subsection{Payment-Division Functions (Seller's Perspective)}\label{sec:payment_division_func}
\subsubsection{Shapley Approximation}\label{sec:shapley_approx}
In our model (see Section \ref{sec:market_dynamics}), a buyer only makes an aggregate payment to the market based on the increase in accuracy experienced (see $\mathcal{RF}^*$ in \eqref{eq:myerson}). It is thus up to the market to design a mechanism to fairly (as defined in Property \ref{prop:fair}) allocate the revenue among the sellers to incentivize their participation. Following the seminal work in~\cite{shapley_allocation}, there have been a substantial number of applications (cf. \cite{powerindices, balkanski2017statistical}) leveraging the ideas in~\cite{shapley_allocation} to fairly allocate cost/reward among strategic entities cooperating towards a common goal. Since the Shapley algorithm stated in \eqref{eq:shapley_mechanism} is the unique method to satisfy Property \ref{prop:fair}, but unfortunately runs in time $\Theta(2^M)$, the best one can do is to approximate  \eqref{eq:shapley_mechanism} as closely as possible. 

In Algorithm \ref{alg-payment-division}, we uniformly sample from the space of permutations over $[M]$ to construct an approximation of the Shapley value in \eqref{eq:shapley_mechanism}. To construct Algorithm \ref{alg-payment-division} more precisely, we need to define some additional quantities. Let $\sigma_{[M]}$ refer to the set of all permutations over $[M]$. For any permutation $\sigma \in \sigma_{[M]}$, let $[\sigma < m]$ refer to the set of features in $[M]$ that came before $m$. 

The key observation in showing that Algorithm \ref{alg-payment-division} is effective, is that instead of enumerating over all permutations in $\sigma_{[M]}$ as in the Shapley allocation, it suffices to sample $\sigma_{k} \in \sigma_{[M]}$ uniformly at random with replacement, $K$ times, where $K$ depends on the $\epsilon$-approximation a practitioner desires. We provide guidance on how to pick $K$ in Section \ref{sec:fairness}. We note some similar sampling based methods, albeit for different applications (cf. \cite{shapley_monte_carlo, ShapleyPolyCastro2009, shapleyApproxSampling2013}). 

In Theorem \ref{thm:shapley_approx}, we show that Algorithm \ref{alg-payment-division} gives an $\epsilon$-approximation for \eqref{eq:shapley_mechanism} with high probability while running in time $O(M^2)$. 
 
\begin{algorithm}[t]
	\SetAlgoNoLine
	\KwIn{$Y_n, \wbX_M, K$}
	\KwOut{$\hat{\psi}_n= [\hat{\psi}_n(m) : m \in [M]]$}
	Let $\mathcal{B}_{\text{net}}(\epsilon)$ be an $\epsilon$-net of $\mathcal{B}$; \\
	\For{$m \in [M]$ } 
	{
		\For{$k \in [K]$}
		{
			$\sigma_k \sim \text{Unif}(\sigma_{[M]})$; \\
			$G = \mathcal{G}(Y_n, \mathcal{M}(\bX_{[\sigma_k < m]})) $; \\
			$G^+ = \mathcal{G}(Y_n, \mathcal{M}(\bX_{[\sigma_k < m \ \cup \ m]})) $; \\
			$\hat{\psi}^k_n(m) = [G^+ - G]$
		}
		$\hat{\psi}_n(m) = \frac{1}{K} \sum^{K}_{k=1}\hat{\psi}^k_n(m)$
	}
	\caption{SHAPLEY-APPROX: $\mathcal{PD}^*_{A}(Y_n, \wbX_M, K)$}
	\label{alg-payment-division}
\end{algorithm}

\subsubsection{Robustness to Replication}
Recall from Section \ref{sec:replication_robustness} that for freely replicable goods such as data, the standard Shapley notion of fairness does not suffice (see Example \ref{example:shapley_not_robust} for how it can lead to undesirable revenue allocations). Though this issue may seem difficult to overcome in general, we again exploit the particular structure of data as a path forward. Specifically, we note that there are {\em standard methods to define the ``similarity" between two vectors of data}. A complete treatment of similarity measures has been done in~\cite{goshtasby2012similarity}. We provide two examples:
\begin{example}
Cosine similarity, a standard metric used in text mining and information retrieval, is:
\[
\frac{|\langle X_1, X_2 \rangle|}{||X_1||_2 ||X_2||_2}, \ X_1, X_2 \in \mathbb{R}^{T}
\] 
\end{example}

\begin{example}
``Inverse" Hellinger distance, a standard metric to define similarity between underlying data distributions, is: $1 -\frac{1}{2} \sum_{x \in \mathcal{X}}(\sqrt{p_1(x)} - \sqrt{p_2(x)})^2)^{1/2}, \ p_1 \sim X_1, \ p_2 \sim X_2$.
\end{example}

We introduce some natural properties any such similarity metric must satisfy for our purposes,
\begin{definition}[\textbf{Adapted from \cite{goshtasby2012similarity}}]\label{def:similarity_metric}
A similarity metric is a function, $\mathcal{SM}: \mathbb{R}^{T} \times \mathbb{R}^{T} \rightarrow [0, 1]$, that satisfies: (i) Limited Range: $0 \le \mathcal{SM}(\cdot, \cdot) \le 1$; (ii) Reflexivity: $\mathcal{SM}(X, Y) = 1$ if and only if $X = Y$; (iii) Symmetry:  $\mathcal{SM}(X, Y) = \mathcal{SM}(Y, X)$; (iv) Define $d\mathcal{SM}(X, Y) = 1 - \mathcal{SM}(X, Y)$, then Triangle Inequality: $d\mathcal{SM}(X, Y) + d\mathcal{SM}(Y, Z) \geq d\mathcal{SM}(X, Z)$
\end{definition}

In Algorithm \ref{alg-payment-division-robust}, we construct a ``robust-to-replication" version of the randomized Shapley approximation algorithm by utilizing Definition \ref{def:similarity_metric} above. 

Intuitively, the algorithm penalizes similar features (relative to the similarity metric, $\mathcal{SM}$) to disincentivize replication. We provide guidance on how to pick the hyper-parameter $\lambda$ in Section \ref{sec:main_results}.  

In Theorem \ref{thm:shapley_robust}, we show Algorithm \ref{alg-payment-division-robust} is $\epsilon$-``Robust to Replication" i.e. Property \ref{prop:replication_robustness} (Robustness-to-Replication) holds. See the example below for an illustration of the effect of Algorithm \ref{alg-payment-division-robust} on undesired replication.

\begin{example}\label{example:shapley_robust}
Recall Example \ref{example:shapley_not_robust} where there are two sellers, A and B, each selling an identical feature. In that example, if Seller A replicated her feature, her Shapley allocation increased from $\frac{1}{2}$ to $\frac{2}{3}$. If we instead apply Algorithm \ref{alg-payment-division-robust} (with $\lambda = 1$), then it is easy to see that her Shapley allocation decreases from  $\frac{1}{2e}$ to $\frac{2}{3e^2}$, ensuring Property \ref{prop:replication_robustness}  holds. See Figure \ref{fig:shapley_robust} for an illustration.
\end{example}

\begin{algorithm}[t]
	\SetAlgoNoLine
	\KwIn{$Y_n, \wbX_M, K, \mathcal{SM}, \lambda$}
	\KwOut{$\psi_n = [\psi_n(m) : m \in [M]]$}
	$\hat{\psi}_n(m)$ = SHAPLEY-APPROX($Y_n, \mathcal{M}, \mathcal{G}, K$); \\
	$\psi_n(m) = \hat{\psi}_n(m)\exp(- \lambda \sum_{j \in [M] \setminus \{m\}} \mathcal{SM}(X_m, X_j))$; \\	
	\caption{SHAPLEY-ROBUST: $\mathcal{PD}^*_{B}(Y_n, \wbX_M, K, \mathcal{SM}, \lambda)$}
	\label{alg-payment-division-robust}
\end{algorithm}

\begin{figure}[]
	\centering
	\includegraphics[width=0.45\textwidth]{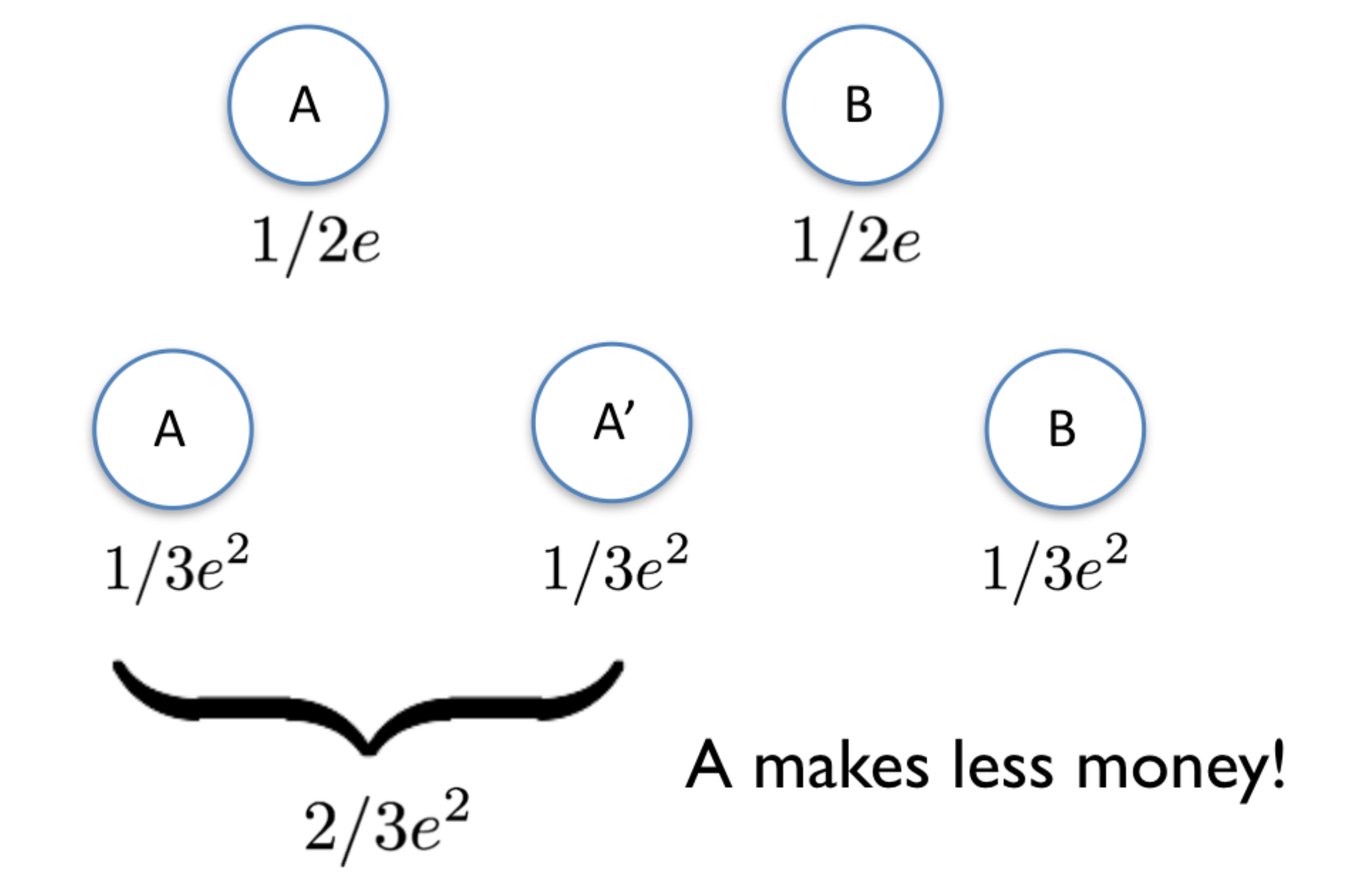}
	\caption{A simple example illustrating how SHAPLEY-ROBUST down weights similar data to ensure robustness to replication.}
	\label{fig:shapley_robust}
\end{figure}

	\section{Main Results} \label{sec:main_results}
\subsection{Assumptions.} \label{sec:key_assumptions}
To give performance guarantees, we state four mild and natural assumptions we need on: (i) $\mathcal{AF}^*$ (allocation function); (ii) $\mathcal{M}$ (ML algorithm); (iii) $\mathcal{RF}^*$ (revenue function); (iv) $b_n$ (bids made).

\begin{assumption}[$\mathcal{AF}^*$ is Monotonic]\label{assumption:robustness_to_noise}
$\mathcal{M}, \mathcal{AF}^*$ are such that an increase in the difference between $p_n$ and $b_n$ leads to a decrease in $\mathcal{G}$ i.e. an increase in ``noise" cannot lead to an increase in prediction accuracy. Specifically, for any $Y_n, p_n$, let $\wbX^{(1)}_M, \wbX^{(2)}_M$ be the outputs of $\mathcal{AF}(p_n, b^{(1)}; \bX_M), \mathcal{AF}(p_n, b^{(2)}; \bX_M)$ respctively. Then if $b^{(1)} \leq b^{(2)}$, we have $\mathcal{G}\Big(Y_n, \mathcal{M}(\wbX^{(1)}_M)\Big) \leq \mathcal{G}\Big(Y_n, \mathcal{M}(\wbX^{(2)}_M)\Big) $.
\end{assumption}
\begin{assumption}[$\mathcal{M}$ is Invariant to Replicated Data]\label{assumption:robustness_to_replication}
$\mathcal{M}$ is such that replicated features do not cause a change in prediction accuracy. Specifically, $\forall \ S \subset [M], \ \forall \ Y_n, \ \forall \ m \in S$, let $m^+_i$ refer to the  $i^{th}$ replicated copy of $m$ (i.e. $X^+_{m, i} = X_m$). Let $S^+ = \cup_m (m \cup_i m^+_i)$ refer to the set of original and replicated features. Then $\mathcal{G}(Y_n, \mathcal{M}(\bX_{S})) = \mathcal{G}(Y_n, \mathcal{M}(\bX_{S^+}))$
\end{assumption}
\begin{assumption}[$\mathcal{RF}^*$ is Lipschitz]\label{assumption:lipschitz_revenue}
The revenue function $\mathcal{RF}^*$ is $\mathcal{L}$-Lipschitz with respect to price. Specifically, for any $Y_n, b_n, p^{(1)}, p^{(2)}$, we have $|\mathcal{RF}^*(p^{(1)}, b_n, Y_n) - \mathcal{RF}^*(p^{(2)}, b_n, Y_n)| \le \mathcal{L} | p^{(1)} - p^{(2)}| $.
\end{assumption}
\begin{assumption}[Bounded Bids]\label{assumption:bounded_bids}
The set of possible bids $b_n$ for $n \in [N]$ come from a closed, bounded set $\mathcal{B}$. Specifically, $b_n \in  \mathcal{B}$, where $\text{diameter}(\mathcal{B}) = D$, where $D < \infty$.
\end{assumption}

\begin{remark}
We provide some justification for Assumptions \ref{assumption:robustness_to_noise} and \ref{assumption:robustness_to_replication} above, which impose requirements of the ML algorithm and the accuracy metric  (i.e. $\mathcal{M}$ and $\mathcal{G}$). These assumptions require that: (i) as more noise is added to data, the less gain in prediction accuracy; (ii) replicated features do not have an effect on the accuracy. Essentially all ML algorithms and the accuracy metrics function like this. Thus these assumptions reflects standard, weak statistical assumptions. Assumptions \ref{assumption:lipschitz_revenue} and \ref{assumption:bounded_bids} are self-explanatory.
\end{remark}

\subsection{Truthfulness.} \label{sec:truthful}
\begin{thm}\label{thm:truthful}
For $\mathcal{AF}^*$, Property \ref{prop:truthful} (Truthfulness) can be achieved if and only if Assumption \ref{assumption:robustness_to_noise} holds. In which case, $\mathcal{RF}^*$ guarantees truthfulness.
\end{thm}
Theorem $\ref{thm:truthful}$ is an application of Myerson's payment function (cf. \cite{myerson1981optimal}) which ensures $b_n = \mu_n$. See Appendix \ref{sec:appendix_truthfulness} for the proof. 

Again, the key is the modeling choice made to define buyer utility as in Definition \ref{def:value_from_prediction}. It lets us parameterize a buyers value for increased accuracy by a scalar, $\mu_n$, which allows us to exploit Myerson's payment function (unfortunately generalization of Myerson's payment function to the setting where $\mu_n$ is a vector are severely limited cf. \cite{Daskalakis:myerson-multi-item}). 

\subsection{Revenue Maximization.} \label{sec:rev_maximizing}
\begin{thm} \label{thm:avg_regret}
Let Assumptions \ref{assumption:robustness_to_noise}, \ref{assumption:lipschitz_revenue} and  \ref{assumption:bounded_bids} hold. Let $p_{n: n \in [N]}$ be the output of Algorithm \ref{alg-price-update-function}. Let $\mathcal{L}$ be the Lipschitz constant of $\mathcal{RF}^*$ (defined as in Assumption \ref{assumption:lipschitz_revenue}). Let $\mathcal{B}_{\max} \in \mathbb{R}$ be the maximum element of $\mathcal{B}$ (where  $\mathcal{B}$ is defined as in Assumption \ref{assumption:bounded_bids}). Then by choosing the algorithm hyper-parameters $\epsilon = (\mathcal{L} \sqrt{N})^{-1}, \ \delta =  \sqrt{\log(|\mathcal{B}_{\text{net}}(\epsilon)|)/N}$, the total average regret is bounded by, 
\[
\frac{1}{N}\Ex[\mathcal{R}(N)] \leq C \mathcal{B}_{\max}
\sqrt{\frac{\log(\mathcal{B}_{\max} \mathcal{L} \sqrt{N} )}{N}} = O(\sqrt{\frac{\log(N)}{N}}),
\]
for some positive constant $C > 0$. Here, the expectation is taken over the randomness in Algorithm \ref{alg-price-update-function}.
Hence, Property \ref{prop:rev_maximizing} (Revenue Maxmization) holds.
\end{thm}

Theorem \ref{thm:avg_regret} proves Algorithm \ref{alg-price-update-function} is a zero regret algorithm. We note the bound is independent of $M$, the number of features sold. See Appendix \ref{sec:appendix_price_update} for the proof. 

As we note in Remark \ref{remark:future_AF_work}, a limitation of the $\mathcal{AF}^*$ we design is that it is fixed, i.e., we degrade each feature by the same scaling. We leave it as future work to design an adaptive $\mathcal{AF}$; instead of fixing $\mathcal{AF}^*$ a priori (as we currently do using standard noising procedures), can we make the noising procedure adaptive to the prediction tasks to further increase the revenue generated (potentially by adding distributional assumptions to the prediction tasks)?


\subsection{Fairness in Revenue Division.} \label{sec:fairness}
\begin{thm}\label{thm:shapley_approx}
Let $\psi_{n, \text{shapley}}$ be the unique vector satisfying Property \ref{prop:fair} (Shapley Fairness) as given in \eqref{eq:shapley_mechanism}.  For Algorithm \ref{alg-payment-division}, pick the following hyperparameter: $K > (M\log(2 / \delta)) / (2 \epsilon^2)$, where $\delta, \epsilon > 0$. Then with probability $1 - \delta$, the output $\hat{\psi_n}$ of Algorithm \ref{alg-payment-division}, achieves the following, 
\[
||\psi_{n, \text{shapley}} - \hat{\psi_n}||_{\infty} < \epsilon.
\]
\end{thm}

Theorem \ref{thm:shapley_approx} gives an $\epsilon$-approximation for $\psi_{n, \text{shapley}}$, the \textit{unique} vector satisfying Property \ref{prop:fair}, in $O(M)$. Recall, computing it exactly would take $\Theta(2^M)$ running time. See Appendix \ref{sec:appendix_fairness} for the proof. 

To the best of our knowledge, the direct application of random sampling to compute feature importances for ML algorithms along with finite sample guarantees is novel. We believe this random sampling method could be used as a model-agnostic tool (not dependent on the particulars of the prediction model used) to assess feature importance - a prevalent question for data scientists seeking interpretability from their prediction models. 

\begin{thm}\label{thm:shapley_robust}
Let Assumption \ref{assumption:robustness_to_replication} hold. For Algorithm \ref{alg-payment-division-robust}, pick the following hyperparameters: $K \ge (M\log(2 / \delta))/(2 (\epsilon / 3)^2), \ \lambda =  \log(2)$, where $\delta, \epsilon > 0$. Then with probability $1 - \delta$, the output, $\psi_n$, of Algorithm \ref{alg-payment-division-robust} is $\epsilon$-``Robust to Replication" i.e. Property \ref{prop:replication_robustness} (Robustness-to-Replication) holds. Additionally Conditions 2-4 of Property \ref{prop:fair} continue to hold for $\psi_n$ with $\epsilon$-precision.
\end{thm}
Theorem \ref{thm:shapley_robust} states Algorithm \ref{alg-payment-division-robust} protects against adversarial replication of data, while maintaining the conditions of the standard Shapley fairness other than balance. 
Again, the key observation, which makes Algorithm \ref{alg-payment-division-robust} possible  is that we can precisely compute similarity between data streams (see Definition \ref{def:similarity_metric}). 
See Appendix \ref{sec:appendix_fairness} for the proof.

A natural question is whether Property \ref{prop:replication_robustness}  and Condition 1 of Property \ref{prop:fair} and can hold together. Unfortunately, as we see from Proposition \ref{prop:shapley_incompatibility}, they cannot (see Appendix \ref{sec:appendix_fairness} for the proof),

\begin{prop}\label{prop:shapley_incompatibility}
If the identities of sellers in the marketplace is anonymized, the balance condition in Property \ref{prop:fair} and Property \ref{prop:replication_robustness} cannot simultaneously hold.
\end{prop}
Note however, Algorithm \ref{alg-payment-division-robust}, down-weights features in a ``local" fashion i.e. highly correlated features are individually down-weighted, while uncorrelated features are not. {\em Hence, Algorithm \ref{alg-payment-division-robust} incentivizes sellers to provide data that is: (i) predictive for a wide variety of tasks; (ii) uncorrelated with other features on sale i.e. has unique information}. 

In Step 2 of Algorithm \ref{alg-payment-division-robust}, we exponentially penalize (i.e. down weight) each feature, for a given similarity metric, $\mathcal{SM}$.  An open question for future work is which revenue division mechanism is the most balanced preserving while being robust to replication? As an important first step, we provide a necessary and sufficient condition for any penalty function
\footnote{We define a general penalty function to be of the form $\hat{\psi}_n(m)f(\cdot)$, instead of $\hat{\psi}_n(m)\exp(- \lambda \sum_{j \in [M] \setminus \{m\}} \mathcal{SM}(X_m, X_j))$ as in Step 2 of Algorithm \ref{alg-payment-division-robust}.}
to be robust to replication for a given similarity metric, $\mathcal{SM}$ (see Appendix \ref{sec:tight_weighting} for the proof), 
\begin{prop}\label{nec_weight}
Let Assumption \ref{assumption:robustness_to_replication} hold. Then for a given similarity metric $\mathcal{SM}$, a penalty function $f$ is ``robust-to-replication"  if and only if it satisfies the following relation for any $c \in \mathbb{Z}_{+}, x \in \Rb_{+}$,
\[
(c+1) f(x+c) \leq f(x)
\]
\end{prop}

\subsection{Efficiency.} \label{sec:efficiency}
\begin{cor}\label{cor:marketplace_is_efficient}
$\mathcal{AF}^*, \mathcal{RF}^*, \mathcal{PF}^*$ run in ${O}(M)$. $\mathcal{PD}^*_{a}$, $\mathcal{PD}^*_{b}$ run in ${O}(M^2)$. Property \ref{prop:efficient} holds.
\end{cor}

\vspace{-0.4em}

See Appendix \ref{sec:appendix_efficiency} for the proof. $\mathcal{AF}^*, \mathcal{RF}^*, \mathcal{PF}^*$ running in ${O}(M)$  is desirable as they need to be re-computed in real-time for every buyer. However, the revenue division algorithms (which run in ${O}(M^2)$) can conceivably run offline as we assume the sellers to be fixed.

	\section{Conclusion} \label{sec:conclusion}
	\noindent \textbf{Modeling contributions.} Our main contribution is a mathematical model for a two-sided data market (Section \ref{sec:model}). We hope our proposed architecture can serve as a foundation to operationalize real-time data marketplaces, which have applicability in a wide variety of important commercial settings (Section \ref{sec:motivating_retail_example}). To further this goal, we define key challenges (Section \ref{sec:marketplace_properties}), construct algorithms to meet these challenges (Section \ref{sec:market_construction}) and theoretically analyze their performance (Section \ref{sec:main_results}). 

To make the problem tractable, we make some key modeling choices. Two of the most pertinent ones include: (i) Buyer $n$'s utility comes solely from the quality of the estimate $\hat{Y}_n$ received, rather than the particular datasets allocated (Definition \ref{def:value_from_prediction}); (ii) the marketplace is allowed to centrally set prices for all features for each buyer, rather than sellers individually setting prices for each feature (Remark \ref{remark:central_price_setting}).

\vspace{0.4em}

\noindent \textbf{Technical contributions.} We highlight two technical contributions, which might be of independent interest.
First, a new notion of ``fairness" required for cooperative games with freely replicable goods (and associated algorithms). As stated earlier (Section \ref{sec:replication_robustness}), such a notion of fairness is especially important in modern applications where users get utility/cost from bundles of digital goods with potentially complex combinatorial interactions (e.g. battery cost attribution for smartphone applications, reward allocation in prediction markets). 
Second, a truthful, zero regret mechanism for auctioning a particular class of combinatorial goods, which utilizes Myerson's payment function and the Multiplicative Weights algorithm. Specifically, if one can find a way of modeling buyer utility/cost through a scalar parameter (e.g. number of unique views for multimedia ad campaigns, total battery usage for smartphone apps), then the framework described can potentially be applied.

\vspace{0.4em}

\noindent \textbf{Future Work.} We reiterate some interesting lines of questioning for future work: (i) how to take into account the externalities of replication experienced by buyers (Remark \ref{remark:externality}); (ii) how to design an adaptive allocation function that further increases revenue generated (Remark \ref{remark:future_AF_work}); (iii) which ``robust-to-replication" revenue division mechanism is the most balanced preserving (Section \ref{sec:fairness})?

\vspace{-0.6em}

	\begin{acks}
	During this work, the authors were supported in part by a MIT Institute for Data, Systems and Society (IDSS) WorldQuant and Thompson Reuters Fellowship.
	\end{acks}

	\bibliographystyle{ACM-Reference-Format}
	\bibliography{bibliography}

	\newpage
	\appendix
	\section{Truthfulness} \label{sec:appendix_truthfulness}

\begin{thm*}[\textbf{Theorem \ref{thm:shapley_approx}}]
For $\mathcal{AF}^*$, Property \ref{prop:truthful} (Truthfulness) can be achieved if and only if Assumption \ref{assumption:robustness_to_noise} holds. In which case, $\mathcal{RF}^*$ guarantees truthfulness.
\end{thm*}

\begin{proof}
This is a classic result from \cite{myerson1981optimal}. We provide the arguments here for completeness and for consistency with the properties and notation we introduce in our work. We begin with the backward direction. By Assumption \ref{assumption:robustness_to_noise} the following then holds $\forall \ b'_n \ge b_n$
\begin{equation}\label{eq:montone_allocation}
\mathcal{G}(Y_n,  \mathcal{M}(\mathcal{AF}^*(b'_n, p_n))) \ge \mathcal{G}(Y_n,  \mathcal{M}(\mathcal{AF}^*(b_n, p_n)))
\end{equation}

To simplify notation, let $h(z; \ \mathcal{G}, p_n, Y_n, \mathcal{M}) = \mathcal{G}(Y_n, \mathcal{M}(\mathcal{AF}^*(z, p_n)))$. In words, $h(z)$ is the gain in prediction accuracy as a function of the bid, $z$, for a fixed $\mathcal{G}, Y_n, \mathcal{M}, p_n$. 

By definition of \eqref{eq:buyer_utility}, it suffices to show that if $b_n \neq \mu_n$, the following holds
\begin{align}
\mu_n \cdot h(\mu_n) - \mu_n \cdot h(\mu_n) + \int_{0}^{\mu_n} h(z) dz \ge  \mu_n \cdot h(b_n) - b_n \cdot h(b_n) + \int_{0}^{b_n} h(z) dz
\end{align}
This is equivalent to showing that
\begin{equation}\label{eq:fairness_proof_key_inequality}
\int_{0}^{\mu_n} h(z) dz \ge  \int_{0}^{b_n} h(z) dz - (b_n - \mu_n) \cdot h(b_n)
\end{equation}

Case 1: $b_n > \mu_n$. In this case, \eqref{eq:fairness_proof_key_inequality} is equivalent to
\begin{align}
(b_n - \mu_n) \cdot h(b_n) \ge \int_{\mu_n}^{b_n} h(z) dz
\end{align}
This is immediately true due to monotonicity of $h(z)$ which comes from \eqref{eq:montone_allocation}. 
Case 2: $b_n < \mu_n$. In this case, \eqref{eq:fairness_proof_key_inequality} is equivalent to
\begin{align}
 \int_{b_n}^{\mu_n} h(z) dz  \ge (\mu_n - b_n) \cdot h(b_n) 
\end{align}
Again, this is immediately true due to monotonicity of $h(z)$. \\

Now we prove the opposite direction, i.e. if we have a truthful payment mechanism, which we denote as $\mathcal{RF}'$,  an increased allocation of features cannot decrease accuracy. Our definition of a truthful payment function implies the following two inequalities $\forall \ b > a$
\begin{align}
a \cdot h(a) - \mathcal{RF}'(\cdot, a, \cdot) &\ge a \cdot h(b) - \mathcal{RF}'(\cdot, b, \cdot) \\
b \cdot h(b) - \mathcal{RF}'(\cdot, b, \cdot) &\ge b \cdot h(a) - \mathcal{RF}'(\cdot, a, \cdot)
\end{align}  
These two inequalities imply
\begin{align}
a \cdot h(a) + b \cdot h(b) &\ge a \cdot h(b) + b \cdot h(a) \implies h(b) (b - a) \ge h(a) (b - a)
\end{align}  
Since by construction $b - a > 0$, we can divide both sides of the inequality by $b-a$ to get
\begin{align}
h(b) \ge h(a) \iff  \mathcal{G}_n(Y_n, \mathcal{M}(\mathcal{AF}^*(b, p_n))) \ge \mathcal{G}_n(Y_n, \mathcal{M}(\mathcal{AF}^*(a, p_n)))
\end{align}  
Since the allocation function $\mathcal{AF}^*(z, p_n)$ is increasing in $z$, this completes the proof.
\end{proof}
	\section{Price Update - Proof of Theorem \ref{thm:avg_regret}} \label{sec:appendix_price_update}

\begin{thm*}[\textbf{Theorem \ref{thm:avg_regret}}]
Let Assumptions \ref{assumption:robustness_to_noise}, \ref{assumption:lipschitz_revenue} and  \ref{assumption:bounded_bids} hold. Let $p_{n: n \in [N]}$ be the output of Algorithm \ref{alg-price-update-function}. Let $\mathcal{L}$ be the Lipschitz constant of $\mathcal{RF}^*$ with respect to price (where $\mathcal{L}$ is defined as in Assumption \ref{assumption:lipschitz_revenue}). Let $\mathcal{B}_{\max} \in \mathbb{R}$ be the maximum element of $\mathcal{B}$ (where  $\mathcal{B}$ is defined as in Assumption \ref{assumption:bounded_bids}). Then by choosing algorithm hyper-parameters $\epsilon = \frac{1}{\mathcal{L} \sqrt{N}}, \ \delta =  \sqrt{\frac{\log(|\mathcal{B}_{\text{net}}(\epsilon)|)}{N}}$ for some positive constant $C > 0$, the total average regret is bounded by, 
\[
\frac{1}{N}\Ex[\mathcal{R}(N)] \leq C \mathcal{B}_{\max}
\sqrt{\frac{\log(\mathcal{B}_{\max} \mathcal{L} \sqrt{N} )}{N}} = O(\sqrt{\frac{\log(N)}{N}}).
\]
where the expectation is taken over the randomness in Algorithm \ref{alg-price-update-function}.
Hence, Property \ref{prop:rev_maximizing} (Revenue Maxmization) holds.
\end{thm*}

\begin{proof}
Our proof here is an adaptation of the classic result from \cite{mw}. We provide the arguments here for completeness and for consistency with the properties and notation we introduce in our work. It is easily seen by Assumption \ref{assumption:robustness_to_noise} that the revenue function $\mathcal{RF}^*$ is non-negative. Now since by construction the gain function, $\mathcal{G} \in [0, 1]$, the range of $\mathcal{RF}^*$ is in $[0, \mathcal{B}_{\max}]$. This directly implies that for all $i$ and $n$, $g^i_n \in [0, 1]$ (recall $g_n^i$ is the (normalized) revenue gain if we played price $i$ for every buyer $n$).

We first prove a regret bound for the best fixed price in hindsight within $\mathcal{B}_{\text{net}}(\epsilon)$. 
Let $g^{\text{alg}}_n$ be the expected (normalized) gain of Algorithm \ref{alg-price-update-function} for buyer $n$. By construction,
\[
g^{\text{alg}}_n = \frac{\sum_{i=1}^{|\mathcal{B}_{\text{net}}(\epsilon)|} w^i_n g^i_n}{W_n}
\]

Observe we have the following inductive relationship regarding $W_n$
\begin{align}
W_{n+1} &= \sum_{i=1}^{|\mathcal{B}_{\text{net}}(\epsilon)|} w^i_{n+1} \\
&= \sum_{i=1}^{|\mathcal{B}_{\text{net}}(\epsilon)|} w^i_n + \delta w^i_n g^i_n \\
&= W_{n} + \delta \sum_{i=1}^{|\mathcal{B}_{\text{net}}(\epsilon)|} w^i_n g^i_n \\
&= W_{n} (1+ \delta g^{\text{alg}}_n) \\
&= W_{1} \Pi_{i=1}^n (1+ \delta g^{\text{alg}}_n) \\
&\stackrel{(a)}= |\mathcal{B}_{\text{net}}(\epsilon)| \cdot \Pi_{i=1}^n (1+ \delta g^{\text{alg}}_n) 
\end{align}
where (a) follows since $W_1$ was initialized to be $|\mathcal{B}_{\text{net}}(\epsilon)|$.

Taking logs and utilizing the inequality $\log(1 + x) \le x$ for $x \ge 0$, we have 
\begin{align}
\log(W_{N+1}) &= \log(|\mathcal{B}_{\text{net}}(\epsilon)|) + \sum_{i=1}^N \log(1 + \delta g^{\text{alg}}_n)) \\
&\le  \log(|\mathcal{B}_{\text{net}}(\epsilon)|) +  \sum_{i=1}^N \delta g^{\text{alg}}_n 
\end{align}

Now using that $\log(1 + x) \ge x - x^2$ for $x \ge 0$, we have for all prices $c^i \in \mathcal{B}_{\text{net}}(\epsilon)$,
\begin{align}
\log(W_{N+1}) &\ge \log(w^i_{n+1}) \\
&= \sum_{n=1}^N \log(1 + \delta g_n^i)) \\
&\ge \sum_{n=1}^N \delta g_n^i - (\delta g_n^i)^2  \\
&\stackrel{(a)}\ge \sum_{i=1}^N \delta g_n^i - \delta^2 N  
\end{align}
where (a) follows since $g^i_n \in [0, 1]$.

Thus for all prices $c^i \in \mathcal{B}_{\text{net}}(\epsilon)$
\[
\sum_{n=1}^N \delta   g^{\text{alg}}_n \ge \sum_{n=1}^N \delta g_n^i  -\log(|\mathcal{B}_{\text{net}}(\epsilon)|) - \delta^2 N \\
\]
Dividing by $\delta N$ and picking $\delta = \sqrt{\frac{\log(|\mathcal{B}_{\text{net}}(\epsilon)|)}{N}}$, we have for all prices $c^i \in \mathcal{B}_{\text{net}}(\epsilon)$
\[
\frac{1}{N}\sum_{n=1}^N g^{\text{alg}}_n \ge \frac{1}{N} \sum_{i=1}^N g_n^i  - 2 \sqrt{\frac{\log(|\mathcal{B}_{\text{net}}(\epsilon)|)}{N}}
\]

So far we have a bound on how well Algorithm \ref{alg-price-update-function} performs against prices in $\mathcal{B}_{\text{net}}(\epsilon)$. We now extend it to all of $\mathcal{B}$. Let $g^{\text{opt}}_n$ be the (normalized) revenue gain from buyer $n$ if we had played the optimal price, $p^*$ (as defined in Property \ref{prop:rev_maximizing}). Note that by Assumption \ref{assumption:bounded_bids}, we have $p^* \in \mathcal{B}$. Then by the construction of $|\mathcal{B}_{\text{net}}(\epsilon)|$, there exists $c^i \in \mathcal{B}_{\text{net}}(\epsilon)$ such that $| c^i - p^*| \le \epsilon$. Then by Assumption \ref{assumption:lipschitz_revenue}, we have that
\[
| g^{\text{opt}}_n -  g^i_n | = \frac{1}{\mathcal{B}_{\max}}| \mathcal{RF}^*(p^*, b_n, Y_n) - \mathcal{RF}^*(c^i, b_n, Y_n) | \le \frac{\mathcal{L} \epsilon}{\mathcal{B}_{\max}}
\]

We thus have 
\[
\frac{1}{N}\sum_{n=1}^N g^{\text{alg}}_n \ge \frac{1}{N} \sum_{i=1}^N g_n^{\text{opt}}  - 2 \sqrt{\frac{\log(|\mathcal{B}_{\text{net}}(\epsilon)|)}{N}} - \frac{\mathcal{L} \epsilon}{\mathcal{B}_{\max}}
\]

Multiplying throughout by $\mathcal{B}_{\max}$, we get
\[
\frac{1}{N}\sum_{n=1}^N \Ex[\mathcal{RF}^*(p_n, b_n, Y_n)] \ge \frac{1}{N} \mathcal{RF}^*(p^*, b_n, Y_n) - 2 \mathcal{B}_{\max} \sqrt{\frac{\log(|\mathcal{B}_{\text{net}}(\epsilon)|)}{N}} - \mathcal{L} \epsilon
\] 
Now setting $\epsilon = \frac{1}{\mathcal{L} \sqrt{N}}$ and noting that $|\mathcal{B}_{\text{net}}(\epsilon)| \le  \frac{3\mathcal{B}_{\max}}{\epsilon}$, for some positive constant $C > 0$, we have 
\[
\frac{1}{N}\sum_{n=1}^N \Ex[\mathcal{RF}^*(p_n, b_n, Y_n)] \ge \frac{1}{N} \mathcal{RF}^*(p^*, b_n, Y_n) - C 
\mathcal{B}_{\max} \sqrt{\frac{\log(\mathcal{B}_{\max} \mathcal{L} \sqrt{N} )}{N}}
\] 

\end{proof}
	\section{Fairness} \label{sec:appendix_fairness}
\begin{thm*}[\textbf{Theorem \ref{thm:shapley_approx}}] 
Let $\psi_{n, \text{shapley}}$ be the unique vector satisfying Property \ref{prop:fair} as given in \eqref{eq:shapley_mechanism}.  For Algorithm \ref{alg-payment-division}, pick the following hyperparameter: $K > \frac{M\log(2 / \delta)}{2 \epsilon^2}$, where $\delta, \epsilon > 0$. Then with probability $1 - \delta$, the output $\hat{\psi_n}$ of Algorithm \ref{alg-payment-division}, achieves the following
\begin{align}
\|\psi_{n, \text{shapley}} - \hat{\psi_n}\|_{\infty} < \epsilon 
\end{align} 
\end{thm*}

\begin{proof}
It is easily seen that $\psi_{n, \text{shapley}}$ can be formulated as the following expectation
\begin{equation}\label{eq:shapley_mechanism_expectation}
\psi_{n, \text{shapley}}(m) = \mathbb{E}_{\sigma \sim \text{Unif}(\sigma_{S_n})} [\mathcal{G}_n(Y_n, \mathcal{M}_n(\bX_{[\sigma < m \ \cup \ m]})) - \mathcal{G}_n(Y_n, \mathcal{M}_n(\bX_{[\sigma < m]})]
\end{equation}

The random variable $\hat{\psi}^k_n(m)$ is distributed in the following manner: 
\begin{align}
\mathbb{P}\Big(\hat{\psi}^k_n(m) = \mathcal{G}_n(Y_n, \mathcal{M}(\bX_{[\sigma_k < m \ \cup \ m]})) - \mathcal{G}_n(Y_n, \mathcal{M}(\bX_{[\sigma_k < m]})); \sigma \in \sigma_{S_n} \Big) = \frac{1}{S_n!}
\end{align}
We then have  
\begin{align}
\mathbb{E}[\hat{\psi}_n(m)] = \frac{1}{K} \sum^{K}_{k=1} \mathbb{E}[\hat{\psi}^k_n(m)] = \psi_{n, \text{shapley}}
\end{align}

Since $\hat{\psi}_n(m)$ has bounded support between $0$ and $1$, and the $\hat{\psi}^k_n(m)$ are i.i.d, we can apply Hoeffding's inequality to get the following bound
\begin{align}
\mathbb{P}\Big(|\psi_{n, \text{shapley}} - \hat{\psi}_n(m)| > \epsilon \Big) < 2\exp(\frac{-2\epsilon^2}{K})
\end{align}
By applying a Union bound over all $m \in S_n \le M$, we have 
\begin{align}
\mathbb{P}\Big( \|\psi_{n, \text{shapley}} - \hat{\psi}_n\|_{\infty} > \epsilon \Big) < 2M\exp(\frac{-2\epsilon^2}{K})
\end{align}
Setting $\delta = 2M\exp(\frac{-2\epsilon^2}{K})$ and solving for $K$ completes the proof.
\end{proof}

\begin{thm*}[\textbf{Theorem \ref{thm:shapley_robust}}] 
Let Assumption \ref{assumption:robustness_to_replication} hold. For Algorithm \ref{alg-payment-division-robust}, pick the following hyperparameters: $K \ge \frac{M\log(2 / \delta)}{2 (\frac{\epsilon}{3})^2}, \ \lambda =  \log(2)$, where $\delta, \epsilon > 0$. Then with probability $1 - \delta$, the output, $\psi_n$, of Algorithm \ref{alg-payment-division-robust} is $\epsilon$-``Robust to Replication" i.e. Property \ref{prop:replication_robustness} (Robustness-to-Replication) holds. Additionally Conditions 2-4 of Property \ref{prop:fair} continue to hold for $\psi_n$ with $\epsilon$-precision.
\end{thm*}

\begin{proof}
To reduce notational overhead, we drop the dependence on $n$ of all variables for the remainder of the proof. Let $S = \{X_1, X_2, \ldots, X_K\}$ refer to the original set of allocated features without replication. Let $S^+ = \{X_{(1, 1)}, X_{(1, 2)}, \ldots, X_{(1, c_1)}, X_{(2, 1)}, \ldots, X_{(K, c_K)}\}$ (with $c_i \in \mathbb{N}$), be an appended version of $S$ with replicated versions of the original features, i.e. $X_{(m, i)}$ is the $(i-1)$-th replicated copy of feature $X_m$.  

Let $\hat{\psi}, \hat{\psi^+}$ be the respective outputs of Step 1 of Algorithm \ref{alg-payment-division-robust} for $S, S^+$ respectively. The total revenue allocation to seller $m$ in the original and replicated setting is given by the following:
\begin{align}
\psi(m) &= \hat{\psi}(m)\exp(- \lambda \sum_{j \in S_m \setminus \{m\} } \mathcal{SM}(X_m, X_j)) \\
\psi^+(m) &= \sum_{i \in c_m} \hat{\psi}^+\Big((m, i)\Big)\exp(- \lambda \sum_{(j, k) \in S_m^+  \setminus \{(m, i)\}} \mathcal{SM}(X_{(m, i)}, X_{(j, k)})) 
\end{align}

For Property \ref{prop:replication_robustness} to hold, it suffices to show that $\psi^+(m) \le \psi(m) + \epsilon$. We have that
\begin{align*}
&\sum_{i \in c_m} \hat{\psi}^+\Big((m, i)\Big)\exp (- \lambda \sum_{(j, k) \in S_m^+  \setminus \{(m, i)\}} \mathcal{SM}(X_{(m, i)}, X_{(j, k)})) \\
&\stackrel{(a)}\le \sum_{i \in c_m}  \hat{\psi}^+\Big((m, i)\Big)\exp(- \lambda \sum_{j \in [c_m] \setminus {i}} \mathcal{SM}(X_{(m, i)}, X_{(m, j)}))\exp(- \lambda \sum_{l \in S_m \setminus \{m\}} \mathcal{SM}(X_{(m, i)}, X_{(l, 1)})) \\
&\stackrel{(b)}=  \sum_{i \in c_m} \hat{\psi}^+\Big((m, i)\Big)\exp(- \lambda (c_m -1)) \exp(- \lambda \sum_{l \in S_m \setminus \{m\}} \mathcal{SM}(X_{(m, i)}, X_{(l, 1)})) \\
&\stackrel{(c)}\le c_m \bigg(\hat{\psi}^+\Big((m, 1)\Big) + \frac{1}{3}\epsilon \bigg) \exp(- \lambda (c_m -1)) \exp(- \lambda \sum_{l \in S_m \setminus \{m\}} \mathcal{SM}(X_{(m, 1)}, X_{(l, 1)})) \\ 
&\le c_m \bigg( \hat{\psi}^+\Big((m, 1)\Big) + \frac{1}{3}\epsilon \bigg) \exp(- \lambda (c_m -1)) \exp(- \lambda \sum_{j \in S_m \setminus \{m\} } \mathcal{SM}(X_m, X_j)) 
\end{align*}
(a) follows since $\lambda, \mathcal{SM}(\cdot) \ge 0$; (b) follows by condition (i) of Definition \ref{def:similarity_metric}; (c) follows from Theorem \ref{thm:shapley_approx}; \\

Hence it suffices to show that $ c_m \bigg( \hat{\psi}^+\Big((m, 1)\Big) + \frac{1}{3}\epsilon \bigg) \exp(- \lambda (c_m -1)) \le \hat{\psi}(m) + \epsilon \ \forall c_m \in \mathbb{N}$. We have 
\begin{align*}
c_m \exp(- \lambda (c_m -1))  \bigg( \hat{\psi}^+\Big((m, 1)\Big) + \frac{1}{3}\epsilon \bigg) &\stackrel{(d)}\le c_m \exp(- \lambda (c_m -1)) \Big(\frac{\psi(m)}{c_m} + \frac{2}{3}\epsilon \Big)\\
&\stackrel{(e)}\le c_m \exp(- \lambda (c_m -1))\Big(\psi(m) + \frac{2}{3}\epsilon \Big) \\
&\stackrel{(f)}\le c_m \exp(- \lambda (c_m -1))\Big(\hat{\psi}(m) + \epsilon \Big) \\
&\stackrel{(g)}\le \Big(\hat{\psi}(m) + \epsilon \Big) 
\end{align*}
where (d) and (f) follow from Theorem \ref{thm:shapley_approx}; (e) follows since $c_m \in \mathbb{N}$; (g) follows since $c_m  \exp(- \lambda (c_m -1)) \le 1 \ \forall c_m \in \mathbb{N}$ by picking $\lambda = \log(2)$.

The fact that Conditions 2-4 of Property \ref{prop:fair} continue to hold for follow  $\psi_n$ with $\epsilon$-precision follow easily from Theorem \ref{thm:shapley_approx} and the construction of $\psi_n$.
\end{proof}

\begin{prop}[\textbf{Proposition \ref{prop:shapley_incompatibility}}]
If the identities of sellers in the marketplace is anonymized, the balance condition in Property \ref{prop:fair} and Property \ref{prop:replication_robustness} cannot simultaneously hold.
\end{prop}

\begin{proof}
We show this through an extremely simple counter-example consisting of three scenarios.  

In the first scenario, the marketplace consists of exactly two sellers, $A , B$, each selling identical features i.e. $X_A = X_B$. By Condition 1 and 2 of Property \ref{prop:fair}, both sellers must receive an equal allocation i.e. $\psi_1(A) = \psi_1(B) = \frac{1}{2}$ for any prediction task. 

Now consider a second scenario, where the marketplace against consists of the same two sellers, $A$ and $B$, but this time seller $A$ replicates his or her feature once and sells it again in the marketplace as $A'$. Since by assumption the identity of sellers is anonymized, to achieve the ``balance" condition in Property \ref{prop:fair}, we require $\psi_2(A) = \psi_2(B) = \psi_2(A') = \frac{1}{3}$. Thus the total allocation to seller $A$ is $\psi_2(A) + \psi_2(A') = \frac{2}{3} > \frac{1}{2} = \psi_1(A)$ i.e. Property \ref{prop:replication_robustness}  does not hold.

Finally consider a third scenario, where the marketplace consists of three sellers $A , B$ and $C$, each selling identical features i.e. $X_A = X_B = X_C$. It is easily seen that to achieve ``balance", we require  $\psi_3(A) = \psi_3(B) = \psi_3(C) = \frac{1}{3}$. 

Since the marketplace cannot differentiate between $A'$ and $C$, we either have balance or Property \ref{prop:replication_robustness} i.e. ``robustness to replication".
\end{proof}
	\section{Efficiency} \label{sec:appendix_efficiency}
\begin{cor*}[\textbf{Corollary~\ref{cor:marketplace_is_efficient}}]
	$\mathcal{AF}^*, \mathcal{RF}^*, \mathcal{PF}^*$ run in ${O}(M)$. $\mathcal{PD}^*_{a}$, $\mathcal{PD}^*_{b}$ run in ${O}(M^2)$ time. Hence, Property \ref{prop:efficient} holds.
\end{cor*}

\begin{proof}
This is immediately seen by studying the four functions: (i) $\mathcal{AF}^*$ simply tunes the quality of each feature $X_j$ for $j \in [M]$, which is a linear time operation in $M$; (ii) $\mathcal{RF}^*$ again runs in linear time as we require a constant number of calls to $\mathcal{G}$ and $\mathcal{M}$; (iii) $\mathcal{PF}^*$ runs in linear time as we call $\mathcal{G}$ and $\mathcal{M}$ once for every price in $\mathcal{B}_{\text{net}}(\epsilon)$; (iv) $\mathcal{PD}^*_a$ has a running time of $\frac{M^2\log(2 / \delta)}{2 \epsilon^2}$ for any level of precision and confidence given by $\epsilon, \delta$ respectively i.e. we require $\frac{M\log(2 / \delta)}{2 \epsilon^2}$ calls to $\mathcal{G}$ and $\mathcal{M}$ to compute the Shapley Allocation for each feature $X_j$ for $j \in [M]$. The additional step in $\mathcal{PD}^*_b$ i.e. Step 2, is also a linear time operation in $M$ (note that the pairwise similarities between $X_i, X_j$ for any $i, j \in [M]$ can be precomputed).
\end{proof}

	\section{Optimal Balance-Preserving, Robust-to-Replication Penalty Functions} \label{sec:tight_weighting}
In this section we provide a necessary and sufficient condition for ``robustness-to-replication" any penalty function $f : \Rb_{+} \rightarrow \Rb_{+}$ must satisfy, where $f$ takes as argument the cumulative similarity of a feature with all other features. In Algorithm~\ref{alg-payment-division-robust}, we provide a specific example of such a penalty function given by exponential down-weighting. We have the following result holds
\begin{prop*}[\textbf{Proposition \ref{nec_weight}}] 
Let Assumption \ref{assumption:robustness_to_replication} hold. Then for a given similarity metric $\mathcal{SM}$, a penalty function $f$ is ``robust-to-replication"  if and only if it satisfies the following relation
\[
(c+1) f(x+c) \leq f(x)
\]
where $c \in \mathbb{Z}_{+}, x \in \Rb_{+}$.
\end{prop*}

\begin{proof}
Consider the case where a certain data seller with feature $X_i$ has original cumulative similarity $x$, and makes $c$ additional copies of its own data. The following relation is both necessary and sufficient to ensure robustness,
\[
\hat{\psi}_i (c+1)f(x+c) \leq \psi_i f(x) \label{shapley_change}
\]

We first show sufficiency. By Assumption \ref{assumption:robustness_to_replication}, the new Shapley value (including the replicated features) for a single feature $X_i$ denoted by $\hat{\psi}$, is no larger than the original Shapley value, $\psi$, for the same feature. Then it immediately follows that $(c+1) f(x+c) \leq f(x)$.

We now show that it is also necessary. We study how much the Shapley allocation changes when only one player duplicates data. The Shapley allocation for feature $X_i$ is defined as 
\[
\psi_i(v) = \sum_{S \subseteq N \setminus \{i\}}\frac{|S|!(|N| - |S| - 1)!}{|N|!}(v(S \cup \{i\}) - v(S))
\]
A key observation to computing the new Shapley value is that $v(S \cup \{i\}) - v(S) \geq 0$ if $i$ appears before all its copies. Define $M$ to be the number of original sellers (without copying) and $c$ are the additional copies. By a counting argument one can show that 
\begin{align*}
\hat{\psi}_i(v) &= \sum_{i=0}^{M-1} \frac{1}{(M+c)!} {M-i+c-1 \choose M-i-1}[v(S \cup \{i\}) - v(S)] \\
&= \sum_{i=0}^{M-1} \frac{M!}{(M+c)!} {M-i+c-1 \choose M-i-1} \frac{1}{M!}[v(S \cup \{i\}) - v(S)] \\
&= \sum_{i=0}^{M-1} \frac{M!}{(M+c)!} {M-i+c-1 \choose c} \frac{1}{M!}[v(S \cup \{i\}) - v(S)] \\
&\leq \frac{M}{M+c} \sum_{i=0}^{M-1}\frac{1}{M!}[v(S \cup \{i\}) - v(S)] \\
&= \frac{M}{M+c} \psi_i(v)
\end{align*}
Observe this inequality turns into an equality when all the original sellers have exactly the same data. We observe that for a large number of unique sellers then copying does not change the Shapley allocation too much $\simeq -c/M$. In fact, this bound tells us that when there are a large number of sellers, replicating a single data set a fixed number of times does not change the Shapley allocation too much, \textit{i.e.}, $\hat{\psi}_i \approx \hat{\psi}_i$ (with the approximation being tight in the limit as $M$ tends to infinity). Therefore, we necessarily need to ensure that 
\[
(c+1) f(x+c) \leq f(x)
\]
\end{proof}

\begin{remark}
If we make the {\em extremely loose} relaxation of letting $c \in \Rb_{+}$ instead of $\mathbb{Z}_{+}$, then the exponential weighting in Algorithm~\ref{alg-payment-division-robust} is minimal in the sense that it ensures robustness with least penalty in allocation. Observe that the penalty function (assuming differentiability) should also satisfy
\begin{align*}
\frac{f(c+x) - f(x)}{c} &\leq -f(x+c) \\
\lim_{c \rightarrow 0^{+}}\frac{f(c+x) - f(x)}{c}&\leq -f(x) \\
f^{'}(x) &\leq -f(x)
\end{align*}
By Gronwall's Inequality we can see that $f(x) \leq C e^{-Kx}$ for suitable $C, K \geq 0$. This suggests that the exponential class of penalty ensure robustness with the ``least'' penalty, and are minimal in that sense.
\end{remark}

\end{document}